\documentclass[sigconf,pbalance]{acmart}

\usepackage[utf8]{inputenc}

\newcommand{\tg}{\mathcal{G}}
\newcommand{\tge}{\mathcal{E}}
\newcommand{\tcore}{\mathcal{C}}

\usepackage{amsmath}
\usepackage{amsthm}
\usepackage{amsfonts}

\usepackage{booktabs}
\usepackage{multirow}
\usepackage{graphicx}
\usepackage{thmtools}		
\usepackage{mleftright}
\usepackage{subcaption}
\usepackage{xcolor}
\usepackage{xspace}

\usepackage[mathic=true]{mathtools}
\newtheorem{definition}{Definition}
\newtheorem{theorem}{Theorem}
\newtheorem{lemma}{Lemma}

\usepackage[english,linesnumbered,vlined,ruled,nokwfunc,algo2e]{algorithm2e}
\DontPrintSemicolon
\SetKwComment{note}{$\triangleright$ }{}
\SetFuncSty{textsc}
\SetCommentSty{textit}
\SetDataSty{texttt}
\SetKwInput{Initial}{Initial}
\SetKwInput{Parameter}{Param.}
\SetKwInput{Input}{Input}
\SetKwInput{Data}{Data}
\SetKwInput{Output}{Output}
\ResetInOut{Result} 
\let\oldnl\nl%
\newcommand{\nonl}{\renewcommand{\nl}{\let\nl\oldnl}}%
\SetKw{KwAnd}{and} %
\SetKw{KwOr}{or}
\SetKw{KwXor}{xor}
\SetKw{KwNot}{not}
\SetKw{Parallel}{parallel}
\SetKw{Return}{return}
\SetKw{Break}{break}

\usepackage{cleveref}
\Crefname{theorem}{Theorem}{Theorems}
\crefname{definition}{Definition}{Definitions}

\usepackage{paralist}

\usepackage{enumitem}

\usepackage{tikz}
\tikzset{
	vertex/.style={circle,thick,draw,minimum size=4mm,inner sep=0pt,font=\footnotesize},
	treevertex/.style={circle,thick,draw,fill=white,minimum size=4mm,inner sep=0.5pt,font=\scriptsize},
	edge/.style={-,thick,font=\footnotesize},
	hledge/.style={-,thick,red,font=\footnotesize},
	hlbedge/.style={-,thick,blue,font=\footnotesize},
	hlgedge/.style={-,thick,Green,font=\footnotesize}
}
\usetikzlibrary{arrows.meta}

\usepackage{appendix}

\newtheorem{appendixlemma}{Lemma}[section]
\newtheorem{appendixtheorem}{Theorem}[section]
\newtheorem{appendixdefinition}{Definition}[section]

\crefname{appendixlemma}{Lemma}{Lemmas}
\crefname{appendixtheorem}{Theorem}{Theorems}
\crefname{appendixdefinition}{Definition}{Definitions}

\copyrightyear{2025}
\acmYear{2025}
\setcopyright{rightsretained}
\acmConference[WSDM '25]{Proceedings of the Eighteenth ACM International Conference on Web Search and Data Mining}{March 10--14, 2025}{Hannover, Germany}
\acmBooktitle{Proceedings of the Eighteenth ACM International Conference on Web Search and Data Mining (WSDM '25), March 10--14, 2025, Hannover, Germany}\acmDOI{10.1145/3701551.3703556}
\acmISBN{979-8-4007-1329-3/25/03}

\begin{document}

\title{An Edge-Based Decomposition Framework for Temporal Networks}

\author{Lutz Oettershagen}
\email{lutz.oettershagen@liverpool.ac.uk}
\orcid{0000-0002-2526-8762}
\affiliation{%
	\institution{University of Liverpool}
	\city{Liverpool}
	\country{UK}
}
\author{Athanasios L. Konstantinidis}
\email{a.konstantinidis@uoi.gr}
\orcid{0009-0001-5566-5187}
\affiliation{%
	\institution{University of Ioannina}
	\city{Ioannina}
	\country{Greece}
}
\author{Giuseppe F. Italiano}
\email{gitaliano@luiss.it}
\orcid{0000-0002-9492-9894}
\affiliation{%
	\institution{LUISS University}
	\city{Rome}
	\country{Italy}
}

\begin{abstract}
A temporal network is a dynamic graph where every edge is assigned an integer time label that indicates at which discrete time step the edge is available.
We consider the problem of hierarchically decomposing the network and introduce an edge-based decomposition framework that unifies the core and truss decompositions for temporal networks while 
allowing us to consider the network's temporal dimension. Based on our new framework, we introduce the $(k,\Delta)$-core and $(k,\Delta)$-truss decompositions, which are generalizations of the classic $k$-core and $k$-truss decompositions for multigraphs. 
Moreover, we show how $(k,\Delta)$-cores and $(k,\Delta)$-trusses can be efficiently further decomposed to obtain spatially and temporally connected components. 
We evaluate the characteristics of our new decompositions and the efficiency of our algorithms. Moreover, we demonstrate how our $(k,\Delta)$-decompositions can be applied to analyze malicious content in a Twitter network to obtain insights that state-of-the-art baselines cannot obtain.
\end{abstract}


\maketitle

\section{Introduction}

Temporal networks are ubiquitous and have gained increasing attention in recent years due to their ability to capture the dynamic nature of real-world systems~\cite{holme2015modern,wang2019time,nicosia2013graph,michail2016introduction,sarpe2021oden,holme2004structure}.
A temporal network consists of a set of nodes and a set of temporal edges. Each temporal edge exists only at a specific point in time, and the network's topology usually changes over time~\cite{holme2015modern,gionis2024mining}. 
In many applications, the temporal aspect is crucial for understanding the evolution and properties of the considered  systems~\cite{oettershagen2022temporal,tgh,enright2018epidemics,gendreau2015time,santoro2022onbra,hogg2012social,mastrandrea2015contact,klimt2004enron,shao2018anatomy}.
\begin{figure}\centering
	\begin{subfigure}{0.48\linewidth}
		\centering
		\includegraphics[width=1\linewidth]{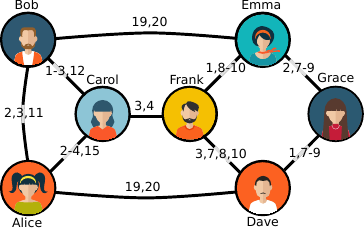}
		\caption{A temporal communication network. The edge labels denote the times of communication.}
		\label{fig:introexample1}
	\end{subfigure}\hfill%
	\begin{subfigure}{0.48\linewidth}
		\centering
		\includegraphics[width=1.0\linewidth]{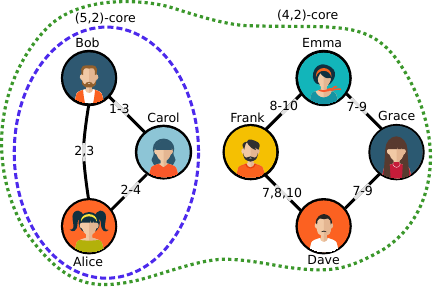}
		\caption{The edge-induced subgraph containing the two most inner cores for $\Delta=2$.}
		\label{fig:introexample2}
	\end{subfigure}
	\caption{Example of our edge-based core decomposition. 
	}
	\vspace{-3mm}
	\label{fig:introexample}
\end{figure}
A common task in data mining is to decompose a network into cohesive components, and the $k$-core and $k$-truss decompositions are useful and popular primitives for this task~\cite{kong2019k,malliaros2020core}.
A \emph{$k$-core} in a static graph $G$ is a maximal subgraph $G_k$ of $G$, such that every node in $G_k$ has at least $k$ neighbors in $G_k$~\cite{seidman1983network}. Moreover, a \emph{$k$-truss} in a static network is an edge-induced subgraph in which each edge is part of at least $k-2$ triangles~\cite{cohen2008trusses}. Both the $k$-core and $k$-truss decompositions find many applications in, e.g., social network analysis~\cite{shao2018anatomy,shin2016corescope}, community detection~\cite{huang2014querying,rossetti2018community}, and network visualization~\cite{alvarez2006k}. However, the static decompositions fail to decompose temporal networks in a useful way as they ignore the temporal aspects. For example, a node or edge that is at one point in time in an inner core or truss can be later at the periphery, i.e., in an outer core or truss, respectively.

In this work, we consider the problem of decomposing a given temporal network into hierarchically arranged 
\emph{temporal components} (cores or trusses). To this end, we first propose an edge-based decomposition framework to capture the temporal dynamics determined by the timestamps of the temporal edges. 
Based on our framework, we introduce the $(k,\Delta)$-core and $(k,\Delta)$-truss decompositions for temporal networks.
Each edge in a $(k,\Delta)$-core appears in the context of $k$ temporal edges with a temporal distance of at most $\Delta$. Specifically, each edge appearing at time $t$ is at both endpoints incident to at least $k$ temporal edges whose timestamps are at most $\Delta$ timesteps before or after time $t$.
Similarly, in an $(k,\Delta)$-truss, each edge appears in the temporal context of $k$ temporal triangles in which the pairwise temporal distances of the edges are upper bounded by $\Delta$.
For example, \Cref{fig:introexample1} shows a toy social network in which the timestamps of interactions between users are shown at each edge. \Cref{fig:introexample2} shows the hierarchical decomposition into two edge-induced temporal components, i.e., the $(5,2)$- and $(4,2)$-cores of the network.

Many previously proposed variants of temporal network decompositions~\cite{galimberti2020span,tgh,wu2015core,bai2020efficient,lotito2020efficient,momin2023kwiq,yang2023scalable,yu2021querying,lin2021mining} show limitations in different ways, e.g., by focusing only on determining static $k$-cores of $k$-trusses in temporal intervals, being not efficient enough for large-scale networks, or being not designed to work on highly dynamic temporal networks. 
Our work aims to overcome these restrictions.

\subsection{Our Contributions} 
\emph{1. Edge-based decomposition framework:}
We introduce a framework for edge-based decompositions of temporal networks. 
Based on this framework, we introduce our new $(k,\Delta)$-core and $(k,\Delta)$-truss decompositions to decompose temporal networks into hierarchically organized temporal cores or trusses, respectively.
We provide efficient algorithms for computing the $(k,\Delta)$-decompositions.

\noindent
\emph{2. Connected components:}
A temporal $(k,\Delta)$-core or $(k,\Delta)$-truss is, in general, not ``connected''. However, there is no standard definition of \emph{connectedness} in temporal graphs~\cite{bhadra2003complexity,casteigts2022simple,costa2023computing,nicosia2012components,kovanen2011temporal}.
We show that a variant of temporal reachability by~\citet{kovanen2011temporal}
can be used to identify \emph{$\Delta$-connected components} of $(k,\Delta)$-cores and $(k,\Delta)$-trusses. 
The decomposition can be computed in linear time and is useful because the connected components of a $(k,\Delta)$-core, or $(k,\Delta)$-truss, are themselves (non-maximal) $(k,\Delta)$-cores, or $(k,\Delta)$-trusses.

\noindent
\emph{3. Evaluation:}
We analyze our algorithms, compare them to existing state-of-the-art baselines, and show that our algorithms can efficiently handle large-scale and highly dynamic temporal networks for which the baselines fail. Finally, we demonstrate in a use-case how the $(k,\Delta)$-core decomposition can be applied to analyze malicious content in a Twitter network to obtain insights that state-of-the-art baselines cannot obtain.

\smallskip
Please refer to~\Cref{appendix:proofs} for the omitted proofs.

\subsection{Related Work}\label{sec:relatedwork}

Comprehensive overviews of temporal graphs are provided by~\cite{holme2015modern,wang2019time}.
Moreover, there are thorough surveys and introductions of core decompositions~\cite{kong2019k,malliaros2020core}.
Our decomposition framework is inspired by \citet{batagelj2011fast} who generalized the static $k$-core decomposition using monotone vertex property functions. 
Recently, there has been an increasing interest in decompositions of temporal networks---
\Cref{tab:overview}  shows an overview of related decompositions.
Wu et al.~\cite{wu2015core}  propose a $(k,h)$-core decomposition where each vertex in a $(k,h)$-core has at least $k$ neighbors, and there are at least $h$ temporal edges to each neighbor.
The authors of~\cite{bai2020efficient} extend the previous approach for maintaining the temporal $(k,h)$-cores under dynamic updates. 
Galimberti et al.~\cite{galimberti2020span} introduced temporal $(k,I)$-span-cores where each core is associated with a time interval $I$ for which the coreness property holds in each time step of the interval $I$. 
\begin{table}%
	\centering
	\caption{Overview of related decompositions, with $n$ and $m$ being the number of nodes and temporal edges, resp., $m_I$ ($m_A$) being the temporal edges in the interval $I$ (the aggregated graph, resp.), $\mathcal{T}$ is the interval spanned by the temporal graph, $\delta_m$ is the max.~degree.}
	\label{tab:overview}
	\resizebox{1\linewidth}{!}{\renewcommand{\arraystretch}{1.1}\setlength{\tabcolsep}{4pt}
		\begin{tabular}{lcll}\toprule
			\textbf{Variant}                     & \quad\textbf{Reference}\quad \mbox{} & \textbf{Running Time}  & \textbf{Description}  \\ \midrule
			Static $k$-core                      & \cite{seidman1983network}  &  $\mathcal{O}(n+m)$                           & static cores  \\ 
			Static $k$-truss                     & \cite{cohen2008trusses}    &  $\mathcal{O}(m^{1.5})$                       & static trusses  \\
			Historical $k$-core                  & \cite{yu2021querying}      & $\mathcal{O}(\log m + m_I)$                   & static cores spanning fixed interval\\
			Time-range $k$-core                  & \cite{yang2023scalable}    & $\mathcal{O}(\log m + |I|\cdot m_I)$          & static cores in fixed interval\\
			$(k,h)$-core                         & \cite{wu2015core}          & $\mathcal{O}(n+m)$                            & parallel temporal~edges  \\
			Span-core                            & \cite{galimberti2020span}  & $\mathcal{O}(|\mathcal{T}|^2\cdot m)$         & cores spanning intervals  \\
			Span-truss                           & \cite{lotito2020efficient} & $\mathcal{O}(|\mathcal{T}|^2 \cdot m^{1.5})$  & trusses spanning intervals  \\
			$(\eta,k)$-pseudocore                & \cite{tgh}                 & $\mathcal{O}(m \eta \cdot \delta_m)$          & based on temporal H-index \\
			$(L,K)$-lasting core                 & \cite{HungT21}             & $\mathcal{O}(nm^2\cdot L)$                    & $k\geq K$-core lasting $L$ steps  \\
			$(l, \delta)$-dense core             & \cite{qin2022mining}       & $\mathcal{O}(n\cdot|\mathcal{T}|)$            & min.~average degree in interval  \\
			$(\mu, \tau, \epsilon)$-stable core  & \cite{qin2020mining}       & $\mathcal{O}(m\cdot m_A)$                     & min.~\# similar neighb. in interval  \\  
			$(\theta,\tau)$-persistent $k$-c.    & \cite{li2018persistent}    & NP-hard                                       & persistence in sliding window  \\ 
			\midrule
			$\mathbf{(k,\Delta)}$\textbf{-core}  &\multirow{2}{*}{\textbf{This work}}&$\mathcal{O}(m\cdot \delta_m)$& based on temporal edge~degree  \\
			$\mathbf{(k,\Delta)}$\textbf{-truss} &&$\mathcal{O}(m\cdot \delta_m^2)$ & based on temporal edge~support  \\
			
			\bottomrule
		\end{tabular}
	}	
\end{table}
There is a quadratic number of time intervals for which a temporal span-core can exist, and the asymptotic running times of the proposed algorithms are in $\mathcal{O}(|\mathcal{T}|^2\cdot |E|)$ where $\mathcal{T}$ is the interval spanned by the temporal graph.
Hung et al.~\cite{HungT21} define a temporal community as a $(L,K)$-lasting core. A $(L,K)$-lasting core is a set of vertices that forms a $k$-core with $k\geq K$ that lasts for time $L$. The definition is similar to the one of the $(k,I)$-spanning core; however, here, the length of the interval $I$ is specified by $L$. Note that for both the $(L,K)$-lasting core and the $(k,I)$-spanning cores, the requirement that the $k$-core has to exist in each time step of the interval is often too restrictive as many real-world temporal networks are layer-wise sparse. 
Qin et al.~\cite{qin2022mining} address temporal communities as $(l, \delta)$-maximal dense core, which requires a core to maintain an average degree of at least $\delta$ throughout a time interval lasting no less than $l$ units. However, their proposed solution is infeasible for networks 
spanning a long interval $\mathcal{T}$
due to the space complexity in $\mathcal{O}(\alpha |\mathcal{T}|+|\tge|)$ with $\alpha$ being the maximum number of nodes in a core. 
\citet{tgh} introduces a (non-hierarchical) decomposition of temporal networks into so-called $(\eta,k)$-pseudocores, describing components with high communication capabilities, where $\eta$ is the depth of a recursive computation of a temporal H-index variant. 
\citet{qin2020mining} explore the concept of stable communities in temporal networks, introducing $(\mu, \tau, \epsilon)$-stable cores. In this context, a node is considered a part of a $(\mu, \tau, \epsilon)$-stable core if it maintains no fewer than $\mu$ neighbors, each exhibiting a similarity of at least $\epsilon$, across at least $\tau$ snapshots within the temporal network.
\citet{yu2021querying} discuss \emph{historical $k$-cores} which are the $k$-cores in the aggregated graphs with respect to time intervals, i.e., given a time interval $I$ the historical $k$-cores wrt.~$I$ are the $k$-cores of the aggregated graph spanned by $I$. Similarly, \citet{yang2023scalable} extend the work of \cite{yu2021querying} and introduce \emph{time-range k-cores queries} by allowing the resulting $k$-cores to be induced by any subinterval $I'\subseteq I$ of the time interval $I$.
The works of~\cite{yu2021querying,yang2023scalable} focus on efficiently answering queries for the standard $k$-core definition in given time intervals, and the authors introduce indexing techniques to answer such queries efficiently. Zhong et al.~\cite{zhong2024unified} propose a framework to unify such $k$-core queries for temporal networks.
In a different direction, the authors of~\cite{momin2023kwiq} introduce the concept of identifying \emph{core-invariant nodes} in temporal networks. A core invariant node keeps a core number above a given threshold within a certain time interval. 
\citet{li2018persistent} define the $(\theta,\tau)$-persistent $k$-core, which requires an intricate \emph{persistency} function to be maximized, leading to NP-hard optimization problem of finding the largest $(\theta,\tau)$-persistent $k$-core.
Finally, a closely related is the concept of the truss-decomposition. A $k$-truss is the maximal edge-induced subgraph in which each edge is part of at least $(k-2)$ triangles~\cite{cohen2008trusses}. The temporal $k$-truss~\cite{lotito2020efficient} is an extension of~\cite{galimberti2020span} using the $k$-truss concept instead of $k$-cores, where a $(k,I)$-span truss is a truss that exists in each time step of the interval $I$. 
However, similar to the span-core, there is a quadratic number of time intervals for which a temporal span-truss can exist leading to a asymptotic running time complexity in $\mathcal{O}(|\mathcal{T}|^2\cdot |E|^{1.5})$ rendering the approach infeasible for many real-world temporal networks.

\section{Preliminaries}\label{sec:preliminaries}

\begin{table}
	\caption{Commonly used notations}
	\label{table:notation}
	\centering%
	\resizebox{1\linewidth}{!}{\renewcommand{\arraystretch}{1.0}%
		\begin{tabular}{l@{\hspace{3mm}}l@{}}%
			\toprule
			\textbf{Symbol} & \textbf{Definition}
			\\\midrule
			$\tg=(V,\tge)$                   & Temporal graph $\tg$ with nodes $V$ and temporal edges $\tge$ \\
			$e=(\{u,v\},t)$                             & Temporal $\{u,v\}$-edge at time $t\in\mathbb{N}$\\
			$m=|\tge|$, $n=|V|$ & Numbers of temporal edges and nodes\\
			$\delta_m$ & Maximum degree in $\tg$ \\   
			
			$\mathcal{T}$ & Time interval spanned by graph \\
			$T(\tg)$ & Set of timestamps in $\tg$, i.e., $\{t \!\mid\!(u,v,t)\in \tge\}$\\
			$\varphi:\tge\times2^\tge\rightarrow~\mathbb{R}$ & Temporal edge weight function \\
			
			$\tcore^{\varphi}_{r}$ & Maximum edge-induced subgraph with $\varphi(e,\tge(\tcore^{\varphi}_{r}))\geq r$ \\
			
			$\Delta\in\mathbb{N}$ & Temporal distance\\
			$\Delta_m\in\mathbb{N}$ & Max. temporal distance of two edges at a node\\
			$d_\Delta(e,\tge')$ & $\Delta$-degree of edge $e\in\tge'$\\
			$s_\Delta(e,\tge')$ & $\Delta$-support of edge $e\in\tge'$\\
			
			$\tcore^{d_\Delta}_{k}$ & $(k,\Delta)$-core \\
			$\tcore^{s_\Delta}_{k}$ & $(k,\Delta)$-truss \\
			$c_\Delta(e)$ & Core number of $e\in \tge$\\
			$\tau_\Delta(e)$ & Truss number of $e\in \tge$\\
			
			$\xi$ & Maximum of $\Delta$-incident edges at any edge $e\in\tge$\\
			\bottomrule
	\end{tabular}}
\end{table}

\Cref{table:notation} in shows an overview of our notation.
A static \emph{(multi-)graph} $G=(V, E)$ consists of a finite set of nodes $V$ and a finite (multi-)set $E\subseteq\{\{u,v\}\subseteq V\mid u\neq v\}$ of undirected edges. 
We say that an edge $e=\{u,v\}$  is \emph{incident} to $u$ and $v$. 
The degree $\delta(v)$ of a node $v\in V$ is the number of edges incident to $v$.
Given a graph $G$ and $k\in\mathbb{N}$, a subgraph $H$ is a \emph{$k$-core} of $G$ if \emph{(i)} each vertex $u \in V(H)$ has degree  at least $k$ in $H$ and \emph{(ii)} $H$ is maximal. 

An \emph{temporal network} (or temporal graph) $\tg=(V, \tge)$ consists of a finite set of nodes $V$ and a finite set $\tge$ of undirected \emph{temporal edges} $e=(\{u,v\},t)$ with $u$ and $v$ in $V$, $u\neq v$, and \emph{timestamp} $t \in \mathbb{N}$. The timestamp specifies when the edge exists in the graph. 
We define $n=|V|$, $m=|\tge|$, and $T(\tg)$=$\{t\!\mid\!(\{u,v\},t)\in \tge\}$.
We use $\tge(\tg)$ to denote the temporal edges of $\tg$.
We assume that $n\leq 2m$, which holds unless isolated vertices exist. 
We denote with $\delta_m$ the maximum degree in $\tg$.
For a subset $\tge'\subseteq \tge$, we define the \emph{edge-induced subgraph} $\tg'=(V',\tge')$ with $V'=\{u,v\mid (\{u,v\},t)\in\tge'\}$; we may also write $\tg'\subseteq \tg$.
Finally, all our new definitions and algorithms can be easily adapted to respect a restrictive time interval $I=[\alpha,\beta]$ with $\alpha,\beta\in\mathbb{N}$ such that only temporal edges $e=(\{u,v\},t)\in\tge$ with $t\in I$ are considered.
We do not make the restrictive time interval explicit for ease of readability.

\section{An Edge-based Decomposition Framework}\label{sec:kdcore}
In this section, we introduce our new decomposition framework. The motivation origins from the following observations:
Given a temporal network $\tg=(V,\tge)$, each temporal edge $e=(\{u,v\},t_e)\in\tge$ has its timestamp $t_e$ determining the time of the existence of the edge. Hence, the temporal edges $\tge$ fully define the temporal dimension and properties of the network, while the nodes $V$ can be considered static objects. 
Hence, by defining a temporal edge-based decomposition framework, we can naturally account for the temporal dimension.
To this end, we first define the temporal edge weight function  $\varphi:\tge\times2^\tge\rightarrow~\mathbb{R}$ such that $\varphi(e,\tge')$ assigns a weight to edge $e\in\tge$ with respect to $\tge'\subseteq \tge$.
Based on $\varphi$, we decompose the temporal network.
\begin{definition}\label{def:kdeltacore_general}
	Given $r\in\mathbb{R}$ and $\varphi:\tge\times2^\tge\rightarrow~\mathbb{R}$, the $(r,\varphi)$-component of a temporal graph $\tg=(V,\tge)$ is the inclusion-maximal edge-induced subgraph $\tcore^{\varphi}_{r}\subseteq \tg$ such that 
	for each temporal edge $e\in\tge(\tcore^{\varphi}_{r})$ it holds $\varphi(e,\tge(\tcore^{\varphi}_{r}))\geq r$.
\end{definition}

The value $c_{\varphi}(e)$ of $e\in \tge$ is the maximum $r\in\mathbb{R}$ such that $e$ is in a $\tcore^{\varphi}_{r}$ component of $\tg$ but not in a $\tcore^{\varphi}_{r'}$ component with $r'>r$.
We call the function $\varphi$ \emph{monotone} if for 
$\tge_1\subseteq \tge_2 \subseteq \tge$ and for $e\in\tge$, $\varphi(e,\tge_1) \leq \varphi(e,\tge_2)$.
For a monotone function $\varphi$, assigning $c_{\varphi}(e)$ to the edges $e\in\tge$ induces a hierarchical decomposition of $\tge$. Moreover, monotonicity allows us to apply a simple edge peeling strategy to compute the decomposition as shown in \Cref{alg:coredecomp1}.

\begin{algorithm2e}[htb]
	\label[algorithm]{alg:coredecomp1}
	\caption{Decomposition Framework}
	\Input{Temporal graph $\tg=(V,\tge)$ and $\varphi$}
	\Output{$c_{\varphi}(e)$ for all $e\in \tge$}
	Initialize $c[e] = \varphi(e,\tge)$ for all $e\in\tge$ and $\tge'= \tge$\;
	\While{$\tge'\neq\emptyset$}{
		Choose $e\in\tge'$ with $c[e]$ minimal\;
		$\tge'\gets \tge' \setminus \{e\}$\;
		\For{all affected $f\in\tge'$ with $c[f]>c[e]$}{
			Update $c[f] \gets \max(c[e], \varphi(f,\tge'))$\;
		} 
	}
	\Return $c[e]$ for all $e\in \tge$
\end{algorithm2e}

\begin{theorem}\label{theorem:general}
	Given a temporal graph $\tg=(V,\tge)$ and a monotone function $\varphi$, \Cref{alg:coredecomp1} correctly computes $c_{\varphi}(e)$ of all edges $e\in\tge$.
\end{theorem}

In the following, we define instances of the function $\varphi$ to develop our new core and truss decompositions as new essential primitives for temporal network analysis. We also provide efficient corresponding variants of \Cref{alg:coredecomp1}. It is important to note that our \Cref{def:kdeltacore_general} is a general framework as any temporal and non-temporal edge property can be used to potentially define additional variants of decompositions, allowing for the exploration of more advanced methods, such as distance-based~\cite{bonchi2019distance} or motif-based decompositions~\cite{sarpe2024scalable} in future work. 

\subsection{Temporal $(k,\Delta)$-Cores}
Next, we establish our specific temporal edge weight function to identify cohesive temporal subgraphs.
To this end, consider a single temporal edge representing an interaction between two vertices. A necessary condition for a temporal edge to be part of an inner core is that it occurs in the context of many other spatially and temporally local interactions. 
To capture this context, we define the degree of a temporal edge as follows.

\begin{definition}\label{def:deltadeg}
	Let $\tg=(V,\tge)$ be a temporal graph, $\Delta\in\mathbb{N}$, we define the $\Delta$-degree $d_\Delta:\tge\times2^\tge\rightarrow~\mathbb{N}$ as
	\begin{align*}
		d_\Delta(e,\tge')=\min (&|\left\{(\{u,w\},t')\in \tge' \mid |t-t'|\leq \Delta \right\}|,\\&|\left\{(\{v,w\},t')\in \tge' \mid |t-t'|\leq \Delta \right\}|).
	\end{align*}
	We denote two edges $e=(\{u,v\},t)$ and $f=(\{u,w\},t')$ with $|t-t'|\leq \Delta$ as $\Delta$-incident each other.
\end{definition}

Note that each edge counts in its own $\Delta$-degree as it is considered incident to itself.

\begin{lemma}\label{lemma:monotone}
	The $\Delta$-degree $d_\Delta$ is monotone.
\end{lemma}
The $\Delta$-degree $d_\Delta$ as the temporal edge weight function $\varphi$ together with \Cref{def:kdeltacore_general} leads our new \emph{$(k,\Delta)$-core decomposition}. 
And with \Cref{lemma:monotone} and \Cref{theorem:general},
\Cref{alg:coredecomp1}  results in the \emph{core numbers} $c_\Delta$ for all edges and the temporal cores $\tcore^{d_\Delta}_{k}\subseteq \tg$. We will discuss a more efficient implementation in \Cref{sec:coreimpl}.
In the following, we call $\tcore^{d_\Delta}_{k}$ a $(k,\Delta)$-core and the edge-induced subgraph containing only edges with core number exactly $k$ the $(k,\Delta)$-shell of $\tg$.

The $(k,\Delta)$-core decomposition is a generalization of the classical static $k$-core decomposition for multigraphs.
To see this, we first define the \emph{edge $k$-core} for static multigraphs as follows. 
Given $k\in\mathbb{N}$, the edge $k$-core of a multigraph $G$ is the inclusion-maximal edge-induced subgraph $G_k$ of $G$ such that each endpoint $u,v$ of each edge $\{u,v\}$ in $G_k$ has at least $k+1$ neighbors. 
Hence, for $k>1$, a subgraph $H$ is a $k$-core if and only if $H$ is an edge $(k-1)$-core. 

\begin{theorem}\label{lemma:eqcorenum}
	Let $\tg=(V,\tge)$ be a temporal network and $D(u)=\max\{|t_1-t_2| \mid (\{u,v\},t_1), (\{u,w\},t_2) \in \tge \}\}$ be the maximum time difference over edges incident to the vertex $u\in V$. Moreover, let $\Delta_m=\max_{u\in V} D(u)$.
	For all $e\in\tge$, the core number $c_{\Delta_m}(e)$ equals the static edge core number $c_s(e)$. 
\end{theorem}

Hence, our new $(k,\Delta)$-core generalizes the conventional static edge $k$-core (and hence the node-based static $k$-core) if we choose $\Delta\geq \Delta_m$.
However, note that by using smaller values of $\Delta<\Delta_m$, increasingly higher time resolutions can be considered, which is impossible for the static (and other temporal) core decompositions.

Moreover, the $(k,\Delta)$-core allows us to identify the hierarchical organization of temporal edges within the network 
due to the following containment property.

\begin{theorem}\label{theorem:containment}
	Let $\tg=(V,\tge)$ be  a temporal network, $k,k',\Delta,\Delta'\in\mathbb{N}$ with $k\leq k'$, and $\Delta \geq \Delta'$.
	Furthermore, let $\tcore^{d_\Delta}_{k}$ and $\tcore^{d_{\Delta'}}_{k'}$ be a $(k,\Delta)$- and a $(k',\Delta')$-core, respectively.
	Then $\tcore^{d_\Delta}_{k} \subseteq \tcore^{d_{\Delta'}}_{k'}$.
\end{theorem}

We give an example of $(k,\Delta)$-cores in \Cref{sec:comparison}.

\subsubsection{Efficient Computation}\label{sec:coreimpl}
Based on the containment property (\Cref{theorem:containment}), 
our $(k,\Delta)$-core decomposition can be efficiently computed for a fixed value of $\Delta$ by adapting the greedy \emph{peeling} algorithm introduced by~\citet{batagelj2003m}.
\Cref{alg:coredecomp} shows our edge-peeling algorithm for computing the $(\Delta,k)$-cores.
The algorithm removes a temporal edge with the lowest $\Delta$-degree in each iteration. To this end it uses three arrays $a_u[e], a_v[e]$, and $d[e]$ to store for each temporal edge $e=(\{u,v\},t)\in\tge$ the current numbers of $\Delta$-incident edges at the endpoints $u$ and $v$ and their minimum, respectively, i.e.,
after initialization (line 1-6), $d[e]$ equals the minimum number of $\Delta$-incident edges at the endpoints of edge $e$.
Let $\{e_1,\ldots,e_m\}$ be the sequence in which the edges are processed by the for loop in line 8. %
In the $i$-th round, \Cref{alg:coredecomp} processes $e_i$ with $d[e_i]$, i.e., the edge with the lowest $\Delta$-degree currently remaining in the graph.
Now, because for $e_i$ the value of $d[e_i]$ is the lowest, each edge $e$ remaining in $\tg$, has at least $d[e]\geq d[e_i]$ $\Delta$-incident edges at both endpoints. Therefore, $e_i$ is part of a maximal edge-induced subgraph in which each edge has at both endpoints at least $d[e_i]$ $\Delta$-incident edges, i.e., $c_\Delta(e_i) = d[e_i]$. 
The loop in line \ref{alg:coredecomp:for2} processes each temporal edge $e_i$ once in order of minimal $\Delta$-degree, and $d[e_i]$ will not be changed after $e_i$ is processed due to line \ref{alg:coredecomp:if1}.
After the loop in line \ref{alg:coredecomp:for2} ends, the algorithm returns $d$, i.e., the core numbers for all $e\in\tge$.
Now let $\xi$ be the maximum of $\Delta$-incident edges at any edge $e\in\tge$.
In each iteration, we may have to update the value $d[f]$ for each of the at most $\xi$ $\Delta$-incident edges of $e$. 
Determining these edges is possible in $\mathcal{O}(\log \delta_m)$ by storing the edges at each vertex in chronologically ordered incidence lists.
Finally, updating the bin position of $f$ takes only constant time. 

\begin{algorithm2e}[htb]
	\label[algorithm]{alg:coredecomp}
	\caption{$(k,\Delta)$-core decomposition}
	\Input{Temporal graph $\tg=(V,\tge)$ and $\Delta\in\mathbb{N}$}
	\Output{Core number $c_\Delta(e)$ for all $e\in \tge$}
	Initialize~$a_u[e]=0$ and $a_v[e]\!=\!0$ for all $e=(\{u,v\},t)\in \tge$\;
	Initialize $d[e]=0$ for all $e\in \tge$ \;%
	\For{$e=(\{u,v\},t)\in \tge$}{\label{alg:coredecomp:for1}
		$a_u[e]\gets |\left\{(\{u,w\},t')\in \tge \mid |t-t'|\leq \Delta \right\}|$\;
		$a_v[e]\gets |\left\{(\{v,w\},t')\in \tge \mid |t-t'|\leq \Delta \right\}|$\;
		$d[e]=\min(a_u[e],a_v[e])$\;
	}
	Bin sort edges $\tge$ in increasing order of $d[e]$\;
	\For{$e=(\{u,v\},t)\in \tge$ in sorted order }{\label{alg:coredecomp:for2}
		\For{ $f=(\{x,w\},t')\in\tge$ with $x\in \{u,v\}$}{
			\If{$|t-t'|\leq \Delta$ \KwAnd $d[f]>d[e]$}{\label{alg:coredecomp:if1}
				$a_x[f]\gets a_x[f]-1$\;
				$d[f]\gets \min(a_x[f],a_w[f])$\;
				Update the bin position of $f$\;%
			}
		}
		remove $e$ from $\tg$\;
	}
	\Return $d$
\end{algorithm2e}

\begin{theorem}\label[theorem]{theorem:kdcore}
	Given a temporal graph $\tg=(V,\tge)$ and $\Delta\in\mathbb{N}$,
	\Cref{alg:coredecomp} computes the $(k,\Delta)$-core numbers of all $e\in\tge$  correctly in $\mathcal{O}(m\cdot \max(\log\delta_m,\xi))$ time and $\mathcal{O}(m)$ space.
\end{theorem}

\subsection{Temporal $(k,\Delta)$-Trusses} 
The $k$-truss in a static network is the maximal subgraph where each edge is part of at least $k-2$ triangles, aiming to enhance cohesiveness compared to the $k$-core by requiring stronger local connectivity.
We now introduce our temporal truss variant by first defining an edge weighting function counting the number of temporally local triangles in which a temporal edge participates.
\begin{definition}\label{def:deltatrianglesupp}
	Let $\tg=(V,\tge)$ be a temporal graph, $\Delta\in\mathbb{N}$, we define the $\Delta$-support $s_\Delta:\tge\times2^\tge\rightarrow~\mathbb{N}$ as
	\begin{align*}
		s_\Delta(e,\tge')=&|\{ \{e_i,e_j\} \mid  e_i=(\{u,w\},t_1),e_j=(\{v,w\},t_2)\in\tge' \\&\quad\text{ with } u\neq v\neq w, |t-t_1|\leq\Delta, |t-t_2|\leq\Delta \\&\quad\text{ and }|t_1-t_2|\leq\Delta\}|.
	\end{align*}
\end{definition} 

\begin{lemma}\label{lemma:trussmonotone}
	The $\Delta$-support $s_\Delta$ is monotone.
\end{lemma}
By using the $\Delta$-support as the temporal edge weight function $\varphi$ in our decomposition framework (\Cref{def:kdeltacore_general}), we obtain our new \emph{$(k,\Delta)$-truss decomposition}. 
Following \Cref{lemma:trussmonotone} and \Cref{theorem:general}, 
\Cref{alg:coredecomp1} with function $s_\Delta$ results in the \emph{truss numbers} $\tau_\Delta$ for all edges and the \emph{temporal trusses} $\tcore^{s_\Delta}_{k}\subseteq \tg$.
In the following, we call $\tcore^{s_\Delta}_{k}$ a $(k,\Delta)$-truss. 
The $(k,\Delta)$-truss decomposition is a generalization of the conventional static $k$-truss decomposition (with the truss numbers shifted by two) in multigraphs by setting $\Delta\geq \Delta_m$ as defined in \Cref{lemma:eqcorenum}.
Furthermore, \Cref{theorem:containment} holds analogously for $(k,\Delta)$-trusses (replace $d_\Delta$ with $s_\Delta$) allowing for a hierarchical decomposition over $k$ as well as $\Delta$.
The $(k,\Delta)$-truss decomposition can be efficiently computed (we provide an algorithm in \Cref{appendix:trussalg}).
\begin{theorem}\label{theorem:kdtruss}
	Given $\tg=(V,\tge)$,  $\Delta\in\mathbb{N}$, and $s_\Delta^{\max}=\max_{e\in \tge}s_\Delta(e,\tge)$,
	we can compute the $(k,\Delta)$-truss numbers of all $e\in\tge$  correctly in $\mathcal{O}(m\cdot \max(\log\delta_m,\xi^2))$ time and $\mathcal{O}(\max(m, s_\Delta^{\max}))$ space. 
\end{theorem}

\subsection{Identifying Connected Components}\label{sec:deltaconnectedness}

After computing the $(k,\Delta)$-cores or $(k,\Delta)$-trusses of a temporal graph $\tg$, 
we consider the $(k,\Delta)$-core or $(k,\Delta)$-truss $\tcore^{\star}_{k}$ with $\star\in\{d_\Delta,s_\Delta\}$
as the by the edges $e$ with core number $c_\Delta(e)\geq k$, or truss number $\tau_\Delta(e)\geq k$, respectively, induced subgraph of $\tg$.
A natural question that arises is if $\tcore^{\star}_{k}$ is \emph{connected}.
However, connectivity in temporal graphs is less clearly defined than in conventional static graphs~\cite{bhadra2003complexity,casteigts2022simple,costa2023computing,nicosia2012components}.
Commonly, temporal connectivity is based on temporal reachability, which itself is defined using \emph{temporal walks}~\cite{bhadra2003complexity,nicosia2012components}. A  temporal walk is a sequence of connected temporal edges with (strictly) increasing timestamps to capture the possible \emph{flow of information}.
The requirement of increasing timestamps leads to inherent non-symmetric temporal walks, reflecting the fact that information cannot flow backward in time.
This non-symmetry (together with non-transitivity) often makes determining temporal connected components a hard problem. 
For this reason, in the context of temporal motif mining, \citet{kovanen2011temporal} loosened the restriction of increasing timestamps and defined \emph{$\Delta$-walks} in terms of \emph{temporal locality} to ensure symmetry and transitivity.

\begin{definition}[\citet{kovanen2011temporal}]
	A {$\Delta$-walk} $\omega$ in a temporal graph $\tg=(V,\tge)$ is a sequence of $\ell\in\mathbb{N}$ temporal edges $\omega=(e_1=(\{v_1,v_2\},t_1),e_2=(\{v_2,v_3\},t_2)\ldots, e_\ell=(\{v_{\ell},v_{\ell+1}\},t_\ell))$ for which $|t_i-t_{i+1}|\leq \Delta$ for all $1\leq i<\ell$, i.e., $e_i$ and $e_{i+1}$ for $1\leq i <\ell$ are $\Delta$-incident.
	Moreover, we say that an edge $e_j\in \tge$ is $\Delta$-reachable from edge $e_i\in \tge$ if there exists a $\Delta$-walk from $e_i$ to $e_j$.
\end{definition}

For $\Delta\geq \Delta_m$ as defined in \Cref{lemma:eqcorenum}, $\Delta$-reachability equals conventional reachability. For smaller $\Delta$ values, $\Delta$-reachability captures temporal locality motivated by the fact that in many real-world scenarios, the significance of past or future events, e.g., human interactions or information dissemination, tends to diminish with time passing by. For example, in a social network, the impact of interactions from several years ago is usually less influential on an individual's current preferences and decisions compared to recent interactions. Similarly, events in the far future are usually less relevant to the current situation than events that will happen soon. 

It is easy to see that the $\Delta$-reachability leads to a decomposition of the temporal network that is an equivalence relation, i.e., satisfying {reflexivity, symmetry}, and {transitivity.}
We define the $\Delta$-connected components of a temporal network $\tg=(V,\tge)$ (or a subgraph like a $(k,\Delta)$-core or $(k,\Delta)$-truss) as the maximal subset $\tge'\subseteq \tge$ of temporal edges such that all edges in $\tge'$ are pairwise $\Delta$-reachable.

Using $\Delta$-connectedness, we can further decompose $(k,\Delta)$-cores or $(k,\Delta)$-trusses into $\Delta$-connected  $(k,\Delta)$-cores or $(k,\Delta)$-trusses.

\begin{theorem}\label{theorem:dccs}
	Let $\mathcal{C}_k^\star$ with $\star\in\{d_\Delta,s_\Delta\}$ be a $(k,\Delta)$-core or $(k,\Delta)$-truss. The $\Delta$-connected components of $\mathcal{C}_k^\star$ are non-inclusion-maximal $(k,\Delta)$-cores or $(k,\Delta)$-trusses, respectively.
\end{theorem}

\Cref{fig:tgexamplef} shows the two $2$-components of the $(2,2)$-core shown in \Cref{fig:tgexampleb}. 
The $\Delta$-connected components can be computed using a recursive algorithm in linear time~\cite{kovanen2011temporal}.
In \Cref{appendix:ccs}, we propose a new and simple linear-time algorithm based on transforming the temporal graph into a static representation.

\section{Comparison of Decompositions}\label{sec:comparison}

\begin{figure*}\centering
	\begin{subfigure}{0.25\linewidth}\centering
		\begin{tikzpicture}
			\node[treevertex,color=black,fill=white] (a) at (0, 1.6) [circle, draw] {$a$};
			\node[treevertex,color=black,fill=white] (c) at (1.6, 1.6) [circle, draw] {$c$};
			\node[treevertex,color=black,fill=white] (b) at (0, 0) [circle, draw] {$b$};
			\node[treevertex,color=black,fill=white] (d) at (1.6, 0) [circle, draw] {$d$};

			\draw[edge] (a) -- (b) node[midway, left] {$1,20$};
			\draw[edge] (b) -- (c) node[pos=0.25,above,sloped] {$3,8$};
			\draw[edge] (a) -- (c) node[midway, above] {$1,22$};
			\draw[edge] (c) -- (d) node[midway, right] {$6,20$};
			\draw[edge] (a) -- (d) node[pos=0.25,above,sloped] {$4,10$};
			\draw[edge] (d) -- (b) node[midway, below] {$6,23$};
		\end{tikzpicture}
		\caption{Temporal graph $\tg$.}
		\label{fig:tgexamplea}
	\end{subfigure}\hfill%
	\begin{subfigure}{0.35\linewidth}\centering
		\begin{tikzpicture}
			\node at (0.2, 1.8) {$(2,2)$-core};
			\node at (2.6, 1.8) {$(1,2)$-shell};

			\node[treevertex,color=black,fill=white] (a1) at (0, 1.2) [circle, draw] {$a$};
			\node[treevertex,color=black,fill=white] (c1) at (1.2, 1.2) [circle, draw] {$c$};
			\node[treevertex,color=black,fill=white] (b1) at (0, 0) [circle, draw] {$b$};
			\node[treevertex,color=black,fill=white] (d1) at (1.2, 0) [circle, draw] {$d$};
			
			\node[treevertex,color=black,fill=white] (a2) at (2.4, 1.2) [circle, draw] {$a$};
			\node[treevertex,color=black,fill=white] (c2) at (3.6, 1.2) [circle, draw] {$c$};
			\node[treevertex,color=black,fill=white] (b2) at (2.4, 0) [circle, draw] {$b$};
			\node[treevertex,color=black,fill=white] (d2) at (3.6, 0) [circle, draw] {$d$};

			\draw[edge] (a1) -- (b1) node[midway, left] {$1$};
			\draw[edge] (b1) -- (c1) node[pos=0.5,above,sloped] {$3,8$};
			\draw[edge] (a1) -- (c1) node[midway, above] {$1$};
			\draw[edge] (c1) -- (d1) node[midway, right] {$6$};
			\draw[edge] (d1) -- (b1) node[midway, below] {$6$};            
			\draw[edge] (a2) -- (b2) node[midway, left] {$20$};
			\draw[edge] (a2) -- (c2) node[midway, above] {$22$};
			\draw[edge] (c2) -- (d2) node[midway, right] {$20$};
			\draw[edge] (a2) -- (d2) node[pos=0.5,above,sloped] {$4,10$};
			\draw[edge] (d2) -- (b2) node[midway, below] {$23$};
		\end{tikzpicture}
		\caption{$(k,\Delta)$-core decomposition for $\Delta=2$.}
		\label{fig:tgexampleb}
	\end{subfigure}\hfill%
	\begin{subfigure}{0.4\linewidth}\centering
		\begin{tikzpicture}
			\node at (0.2, 1.8) {$(3,5)$-core};
			\node at (2.6, 1.8) {$(2,5)$-shell};
			\node at (5, 1.8) {$(1,5)$-shell};
			\node[treevertex,color=black,fill=white] (a1) at (0, 1.2) [circle, draw] {$a$};
			\node[treevertex,color=black,fill=white] (c1) at (1.2, 1.2) [circle, draw] {$c$};
			\node[treevertex,color=black,fill=white] (b1) at (0, 0) [circle, draw] {$b$};
			\node[treevertex,color=black,fill=white] (d1) at (1.2, 0) [circle, draw] {$d$};
			
			\node[treevertex,color=black,fill=white] (a2) at (2.4, 1.2) [circle, draw] {$a$};
			\node[treevertex,color=black,fill=white] (c2) at (3.6, 1.2) [circle, draw] {$c$};
			\node[treevertex,color=black,fill=white] (b2) at (2.4, 0) [circle, draw] {$b$};
			\node[treevertex,color=black,fill=white] (d2) at (3.6, 0) [circle, draw] {$d$};
			
			\node[treevertex,color=black,fill=white] (a3) at (4.8, 1.2) [circle, draw] {$a$};
			\node[treevertex,color=black,fill=white] (d3) at (6, 0) [circle, draw] {$d$};

			\draw[edge] (a1) -- (b1) node[midway, left] {$1$};
			\draw[edge] (b1) -- (c1) node[pos=0.27,above,sloped] {$3,8$};
			\draw[edge] (a1) -- (c1) node[midway, above] {$1$};
			\draw[edge] (c1) -- (d1) node[midway, right] {$6$};
			\draw[edge] (a1) -- (d1) node[pos=0.75,above,sloped] {$4$};
			\draw[edge] (d1) -- (b1) node[midway, below] {$6$};            
			\draw[edge] (a2) -- (b2) node[midway, left] {$20$};
			\draw[edge] (a2) -- (c2) node[midway, above] {$22$};
			\draw[edge] (c2) -- (d2) node[midway, right] {$20$};
			\draw[edge] (d2) -- (b2) node[midway, below] {$23$};
			
			\draw[edge] (a3) -- (d3) node[pos=0.5,above,sloped] {$10$};
		\end{tikzpicture}
		
		\caption{$(k,\Delta)$-core decomposition for $\Delta=5$.}
		\label{fig:tgexamplec}
	\end{subfigure}
	
	\vspace{3mm}
	\begin{subfigure}{0.33\linewidth}\centering
		\begin{tikzpicture}
			\node at (0.2, 1.8) {$(2,5)$-truss};

			\node[treevertex,color=black,fill=white] (a1) at (0, 1.2) [circle, draw] {$a$};
			\node[treevertex,color=black,fill=white] (c1) at (1.2, 1.2) [circle, draw] {$c$};
			\node[treevertex,color=black,fill=white] (b1) at (0, 0) [circle, draw] {$b$};
			\node[treevertex,color=black,fill=white] (d1) at (1.2, 0) [circle, draw] {$d$};

			\draw[edge] (a1) -- (b1) node[midway, left] {$1$};
			\draw[edge] (b1) -- (c1) node[pos=0.26,above,sloped] {$3$};
			\draw[edge] (a1) -- (c1) node[midway, above] {$1$};
			\draw[edge] (c1) -- (d1) node[midway, right] {$6$};
			\draw[edge] (a1) -- (d1) node[pos=0.25,above,sloped] {$4$};
			\draw[edge] (d1) -- (b1) node[midway, below] {$6$};            
		\end{tikzpicture}
		\caption{$(k,\Delta)$-truss decomposition for $\Delta=5$.}
		\label{fig:tgexampled}
	\end{subfigure}%
	\begin{subfigure}{0.33\linewidth}\centering
		\begin{tikzpicture}
			\node at (1.0, 1.8) {$\Delta$-ccs of the $(2,2)$-core};

			\node[treevertex,color=black,fill=white] (a1) at (0, 1.2) [circle, draw] {$a$};
			\node[treevertex,color=black,fill=white] (c11) at (1.2, 1.2) [circle, draw] {$c$};
			\node[treevertex,color=black,fill=white] (b11) at (0, 0) [circle, draw] {$b$};    \node[treevertex,color=black,fill=white] (c12) at (3.2, 1.2) [circle, draw] {$c$};
			\node[treevertex,color=black,fill=white] (b12) at (2.0, 0) [circle, draw] {$b$};
			\node[treevertex,color=black,fill=white] (d1) at (3.2, 0) [circle, draw] {$d$};

			\draw[edge] (a1) -- (b1) node[midway, left] {$1$};
			\draw[edge] (b11) -- (c11) node[pos=0.5,below,sloped] {$3$};
			\draw[edge] (b12) -- (c12) node[pos=0.5,above,sloped] {$8$};
			\draw[edge] (a1) -- (c1) node[midway, above] {$1$};
			\draw[edge] (c12) -- (d1) node[midway, right] {$6$};
			\draw[edge] (d1) -- (b12) node[midway, below] {$6$};            
		\end{tikzpicture}
		
		\caption{$2$-connected components of the $(2,2)$-core.}
		\label{fig:tgexamplef}
	\end{subfigure}
	
	\caption{Examples of the $(k,\Delta)$-core and $(k,\Delta)$-truss decompositions and $\Delta$-connected components.}
\end{figure*}

We use a toy communication network to compare our new $(k,\Delta)$-decompositions to the static variants and state-of-the-art temporal core and truss decompositions.
\Cref{fig:tgexamplea} shows the temporal graph in which nodes communicate at different times, shown at the edges. 
Each timestamp corresponds to a temporal edge.

\Cref{fig:tgexampleb} shows the $(2,2)$-core and the $(1,2)$-shell. Similarly, \Cref{fig:tgexamplec} shows the $(3,5)$-core and the $(2,5)$- and $(1,5)$-shells. Each temporal edge in the $(3,5)$-core has a $5$-degree of at least three. For example, both endpoints of the temporal edge $(\{a,b\},1)$, are incident to at least three edges $(\{u,v\},t)$ such that $|t-1|\leq 5$. In the case of $a$, these edges are $(\{a,b\},1)$, $(\{a,c\},1)$, and $(\{a,d\},4)$. And for endpoint $b$, the $\Delta$-incident edges are $(\{a,b\},1)$, $(\{b,c\},3)$, and $(\{b,d\},6)$. Similarly, in \Cref{fig:tgexampled}, the $(2,5)$-truss of $\tg$ is shown in which each edge has a $5$-support of at least two.
The $(3,5)$-core reflects pair-wise communications of all four nodes occurring in the interval $[1,8]$, whereas the $(2,5)$-shell shows (non-pairwise) communications occurring in the later interval $[20,23]$.
Note that the $(2,5)$-\emph{core} consist of the $(3,5)$-core together with the $(2,5)$-shell.
\Cref{fig:tgexamplef} shows the $\Delta$-connected components of the $(2,2)$-core, highlighting
the communications between the nodes $\{a,b,c\}$ and $\{b,c,d\}$ in the two distinct intervals $[1,3]$ and $[6,8]$.
It is essential to mention that for our $(k,\Delta)$-core decomposition, one of the key differences to traditional static and temporal core decompositions is that our approach decomposes the network on an edge basis instead of node-wise, often leading to a more fine-grained decomposition. 
We now compare our approach to other decompositions.

\textbf{Static cores and trusses:} A natural question is if the temporal information in the temporal network $\tg$ is necessary or if a purely static network decomposition suffices. Using the conventional static $k$-core decomposition computed on the underlying aggregated graph (or underlying multigraph) results in each node having a core number of three (or six, resp.), showing that the static approach ignores the network's temporal dimension and does not suffice.
Similar arguments hold for the static truss decomposition; for example, the underlying aggregated graph is a static 4-truss.

\textbf{Interval queries:} The approaches that query a static $k$-core during an interval~\cite{yu2021querying,yang2023scalable} can, e.g., identify the $(3,5)$-core shown in \Cref{fig:tgexamplec} by carefully choosing the interval $[1,8]$. However, it is not possible to obtain the, e.g., $(2,2)$-core using the interval $[1,8]$ as this would lead to the inclusion of the temporal edge $(\{a,d\},4)$ as temporal (non-)locality is not considered in the subgraph.

\textbf{$\mathbf{(k,h)}$-cores:} For the $(k,h)$-core decomposition~\cite{wu2015core}, each vertex in a $(k,h)$-core has at least $k$ neighbors, and there are at least $h$ temporal edges to each neighbor.
Note that the actual timestamps or their relative distances to each other are not considered. Hence the temporal network $\tg$ itself is a $(k,h)$-core for $k=3$ and $h=2$.

\textbf{Span cores and trusses:}
For the spanning decompositions (i.e., span-core~\cite{galimberti2020span}, span-truss~\cite{lotito2020efficient}, $(L,K)$-lasting core~\cite{HungT21}) the cores (trusses, resp.) need to exist for some interval $I$ requiring all edges of the core to exist at each timestamp of $I$. This requirement is, in many cases, too restrictive, as we see in the example graph (\Cref{fig:tgexamplea}), which does not contain any non-trivial spanning core (truss) even for the minimum interval length of $|I|=1$.

\smallskip

In conclusion, using the simple temporal network $\tg$ shown in \Cref{fig:tgexamplea}, we can already verify that the discussed baselines lead to significantly different or trivial decompositions. Similar arguments or examples are also possible for other approaches.
As most baselines are node-based, they are usually unable to achieve the fine-grained decomposition on an edge basis as our approach does.
Finally, in \Cref{sec:usecase}, we discuss an empirical use-case of analyzing malicious retweets in a subnetwork of the Twitter graph, and we show that our $(k,\Delta)$-decompositions lead to insights that cannot be obtained using the state-of-the-art baselines.

\section{Experiments}\label{sec:experiments}
We compare our new decompositions with state-of-the-art baselines and 
discuss an use-case using the $(k,\Delta)$-core and $(k,\Delta)$-truss decompositions for analyzing malicious tweets.

\subsection{Experimental Setup}

All experiments run on a computer cluster. Each experiment had an exclusive node with an Intel(R) Xeon(R) Gold 6130
CPU @ 2.10GHz and 96~GB of RAM. We used a time limit of 12 hours.
We implemented our algorithms in C++ using GNU CC Compiler~11.4.0 with the flag \texttt{--O3}\footnote{The source code is available at~\textcolor{black}{\url{https://gitlab.com/tgpublic/tgkd}}.}.
We denote our implementation of \Cref{alg:coredecomp} as \texttt{$(k,\Delta)$-Core} and our implementation of our $(k,\Delta)$-truss algorithm (see \Cref{appendix:trussalg}) as \texttt{$(k,\Delta)$-Truss}.

\subsubsection{Baselines} We use the following state-of-the-art core decomposition baselines:
\begin{itemize}[leftmargin=5.5mm]
	\item \texttt{Stat-$k$-C} is the static $k$-core decomposition algorithm~\cite{batagelj2003m} in which we ignore all time stamps and $\tg$ is interpreted as undirected multilayer graph.
	\item \texttt{$(k,h)$-C} is the $(k,h)$-core decomposition~\cite{wu2015core}. 
	We use the implementation provided by~\cite{oettershagen2022tglib} and $h\in\{2,4,8\}$. 
	\item \texttt{PC-$\eta$} is the $(\eta,k)$-pseudocore decomposition~\cite{tgh}. We use the implementation provided by the authors and $\eta\in\{8,16,32\}$.
	\item \texttt{SpanC} and \texttt{SpanT} are the maximal~span-core and span-truss decompositions~\cite{galimberti2020span,lotito2020efficient}. We use the implementations provided by the authors.
	\item \texttt{$(L,K)$-C} is the algorithm for the lasting $k$-core with $k\geq K$~\cite{HungT21}. 
	This algorithm only returns the $k$-cores with $k\geq K$ and a maximal number of nodes lasting for exactly a time duration of length $L$. We set $L=1$ minute and $k=2$. 
	We use the implementation provided by~\cite{oettershagen2022tglib}.
	\item \texttt{$(l, \delta)$-C} is the algorithm for the maximal dense core~\cite{qin2022mining}. We use the implementation provided by the authors and the proposed default values for the parameters of $l=3$ and $\delta=3$.
\end{itemize}

\subsubsection{Data sets} We use eight real-world network data sets of different sizes and from various domains.
\Cref{table:datasets_stats2} gives an overview. 
Further details are provided in \Cref{appendix:dataset}.

\begin{table}%
	\centering
	\caption{Statistics of the data sets.}  
	\label{table:datasets_stats2}
	\resizebox{1\linewidth}{!}{ \renewcommand{\arraystretch}{1} %
		\begin{tabular}{llrrrcr}\toprule
			&\textbf{Data~Set}     & $|V(\tg)|$ & $|E(\tg)|$ & \textbf{Span} & \textbf{Domain} & \textbf{Ref.} \\ 
			\cmidrule(lr){1-7} \multirow{4}{*}{\rotatebox{90}{\textsf{Small}}}
			&\emph{FacebookMsg}    &   1.9K &  59.8K & 194 days &     social network & \cite{panzarasa2009patterns} \\
			&\emph{Enron}          &  86.8K &   1.1M &  4 years &      email network &        \cite{klimt2004enron} \\ 
			&\emph{AskUbuntu}      & 134.0K & 257.3K &  7 years & question answering &   \cite{paranjape2017motifs} \\
			&\emph{Twitter}        & 346.1K &   2.1M & 176 days &           retweets &       \cite{shao2018anatomy} \\ 
			\cmidrule(lr){1-7} \multirow{4}{*}{\rotatebox{90}{\textsf{Large}}}
			&\emph{Wikipedia}      &   1.9M &  40.0M &  6 years &         co-editing &   \cite{paranjape2017motifs} \\
			&\emph{StackOverflow}  &   2.6M &  48.0M &  8 years & question answering &     \cite{mislove2009online} \\ 
			&\emph{Reddit}         &   3.0M &  84.3M &  9 years &     social network &     \cite{hessel2016science} \\
			&\emph{Bitcoin}        &  48.1M & 111.0M &  7 years &          financial &   \cite{Kondor-2014-bitcoin} \\
			\bottomrule
		\end{tabular}
	}
\end{table}

\begin{table*}%
	\centering
	\caption{Statistics for the chosen values of $\Delta$.}  
	\label{table:props}
	\resizebox{1\linewidth}{!}{ \renewcommand{\arraystretch}{1} 
		\begin{tabular}{lrrrrrrrrrrrrrrrrrrrr}\toprule
			&\multicolumn{4}{c}{\textbf{Duration}}&\multicolumn{4}{c}{\textbf{Avg. $\Delta$-degree}}&\multicolumn{4}{c}{\textbf{Max. $\Delta$-degree}}&\multicolumn{4}{c}{\textbf{Avg. $\Delta$-support}}&\multicolumn{4}{c}{\textbf{Max. $\Delta$-support}}\\\cmidrule(lr){2-5}\cmidrule(lr){6-9}\cmidrule(lr){10-13}\cmidrule(lr){14-17}\cmidrule(lr){18-21}
			\textbf{Data~Set}     & $\Delta_{10\%}$ & $\Delta_{25\%}$ & $\Delta_{50\%}$ & $\Delta_{75\%}$ & $\Delta_{10\%}$ & $\Delta_{25\%}$ & $\Delta_{50\%}$ & $\Delta_{75\%}$ & $\Delta_{10\%}$ & $\Delta_{25\%}$ & $\Delta_{50\%}$ & $\Delta_{75\%}$ & $\Delta_{10\%}$ & $\Delta_{25\%}$ & $\Delta_{50\%}$ & $\Delta_{75\%}$ & $\Delta_{10\%}$ & $\Delta_{25\%}$ & $\Delta_{50\%}$ & $\Delta_{75\%}$   \\ 
			\midrule
			\emph{FacebookMsg}    & 34.0s &  1.6m &  9.6m &  3.4h & 1.12 & 1.41 & 3.13 & 10.46 &  28 &  28 &  30 &   141 & 0.0003 &  0.004 & 0.12 &  2.11 &   2 &   12 &  137 &  1285 \\
			\emph{Enron}          &  9.0m &  1.0h & 10.0h &  2.8d & 1.20 & 1.73 & 4.16 & 12.42 & 229 & 554 & 727 &  1255 &   0.07 &   0.35 & 2.81 & 22.96 & 452 & 1237 & 1962 & 13079 \\ 
			\emph{AskUbuntu}      & 12.4m &  1.5h &  1.1d & 16.9d & 1.02 & 1.13 & 1.54 &  3.52 &  13 &  27 &  31 &   272 & 0.0002 &  0.003 & 0.02 &  0.26 &   1 &    7 &   14 &   145 \\
			\emph{Twitter}        &  5.0s & 21.0s &  6.4m & 18.7h & 1.03 & 1.17 & 1.79 &  5.54 &   7 &  22 & 218 & 14459 & 0.0000 & 0.0001 & 0.01 &  2.83 &   1 &   15 &  531 & 16017 \\ 
			\emph{Wikipedia}      &  1.0d &  2.0d &  7.0d & 31.0d & 5.23 & 6.37 & 9.82 & 17.78  & 1043 & 1638 & 4472 & 16881 & 0.19 & 0.23 & 0.35 & 0.75 & 731 & 731 & 804 & 1530 \\
			\emph{StackOverflow}  &  1.8m &  5.9m & 39.4m & 18.2h & 1.07 & 1.33 & 2.41 & 5.35 & 13 & 21 & 37 & 183 & 0.0049 & 0.04 & 0.25 & 0.60  & 9 & 30 & 203 & 600 \\ 
			\emph{Reddit}         &  1.1m &  4.4m & 32.0m & 10.2h & 1.03 & 1.23 & 2.26 & 6.22 & 53& 63&240 &3146 & 0.0004 & 0.0066 & 0.10 & 0.78 & 10 & 86 & 1260 & 9564 \\
			\emph{Bitcoin}        &  8.3m & 21.9m &  1.8h & 19.8h & 1.79 & 2.32 & 5.16 & 27.01 & 611 & 1572 & 3917 & 17289 & 0.12 & 0.25 & 1.59 & 68.87 & 213 & 711 & 5259 & 111904 \\
			\bottomrule
		\end{tabular}
	}
\end{table*}

\begin{table*}[htb]
	\centering
	\caption{The running times in seconds (OOT--out of time (12h), OOM--out of memory (96GB)).}  
	\label{table:kdcorerunningtimes2}
	\resizebox{1\linewidth}{!}{ \renewcommand{\arraystretch}{1} %
		\begin{tabular}{lrrrrrrrrrrrrrrrrrrr}\toprule
			&\multicolumn{11}{c}{\textbf{Baselines}}&\multicolumn{4}{c}{\textbf{{\texttt{$(k,\Delta)$-Core}}}}&\multicolumn{4}{c}{\textbf{{\texttt{$(k,\Delta)$-Truss}}}}\\
			\cmidrule(lr){2-12}\cmidrule(lr){13-16}\cmidrule(lr){17-20}
			\textbf{Data~Set} & \texttt{Stat-$k$-C} & \texttt{$(k,2)$-C}  &  \texttt{$(k,4)$-C} & \texttt{$(k,8)$-C} & \texttt{PC-$8$} & \texttt{PC-$16$} & \texttt{PC-$32$} &\texttt{SpanC} & \texttt{SpanT} &\texttt{$(L,K)$-C} &\texttt{$(l, \delta)$-C}
			& $\Delta_{10\%}$ & $\Delta_{25\%}$ & $\Delta_{50\%}$ & $\Delta_{75\%}$ & $\Delta_{10\%}$ & $\Delta_{25\%}$ & $\Delta_{50\%}$ & $\Delta_{75\%}$
			\\ \midrule
			\emph{FacebookMsg}   &   0.02 &   0.01 &   0.01 &   0.01 &    0.28 &    0.59 &    1.23 & OOM &    81.28 &    0.22 & OOM &   0.04 &   0.05 &   0.05 &    0.07 &   0.04 &   0.04 &   0.05 & 0.17 \\
			\emph{Enron}         &   0.78 &   0.32 &   0.27 &   0.25 &   37.21 &   79.64 &  166.57 & OOM & 39398.65 &    6.93 & OOM &   1.76 &   2.01 &   2.71 &    4.88 &   2.16 &   2.70 &   5.90 & 29.06 \\
			\emph{AskUbuntu}     &   0.25 &   0.10 &   0.09 &   0.09 &    0.90 &    1.78 &    3.53 & OOM &      OOT &    1.43 & OOM &   0.21 &   0.23 &   0.27 &    0.40 &   0.21 &   0.22 &   0.25 & 0.60 \\
			\emph{Twitter}       &   2.23 &   1.19 &   1.06 &   1.05 &  684.04 & 3521.95 & 7931.67 & OOM & 42039.00 &   12.45 & OOM &   2.85 &   3.02 &   4.10 &   99.74 &   2.83 &   3.06 &   5.27 & 310.26 \\        
			\emph{Wikipedia}     &  98.10 &  50.00 &  45.18 &  44.64 & 1758.70 & 3807.09 & 8367.10 & OOM &      OOT & 1605.45 & OOM & 222.14 & 280.05 & 466.25 &  899.00 & 306.85 & 429.42 & 914.14 & 3192.34 \\
			\emph{StackOverflow} & 122.23 &  73.10 &  55.18 &  52.23 & 2146.92 & 4518.03 & 9694.15 & OOM &      OOT &  263.71 & OOM &  83.29 &  92.65 & 111.19 &  152.80 &  83.86 &  90.31 & 103.18 & 161.23 \\ 
			\emph{Reddit}        & 255.62 & 145.96 & 126.69 & 121.05 & 1696.92 & 3593.85 &     OOT & OOM &      OOT &  498.25 & OOM & 159.74 & 178.69 & 226.34 &  515.22 & 158.80 & 167.79 & 209.99 & 769.91 \\
			\emph{Bitcoin}       & 291.67 & 119.20 & 108.12 & 107.36 &     OOT &     OOT &     OOT & OOM &      OOT & 2410.94 & OOM & 222.35 & 271.93 & 550.86 & 3867.92 & 205.20 & 282.98 &1342.53 & OOT  \\
			\bottomrule
		\end{tabular}
	}
\end{table*}

\begin{figure*}\centering
	\begin{subfigure}{.23\linewidth}\centering
		\includegraphics[width=1\linewidth]{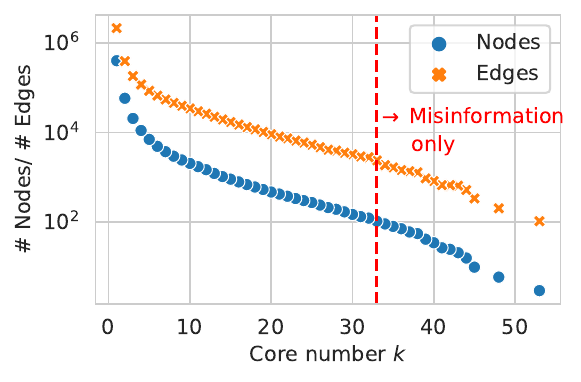}
		\caption{Number of users (nodes) and retweets (edges) in the $(k,\Delta_{50\%})$-cores. \hfill\mbox{}}
		\label{fig:twitter_sizes_h}
	\end{subfigure}\hfill%
	\begin{subfigure}{.23\linewidth}\centering
		\includegraphics[width=1\linewidth]{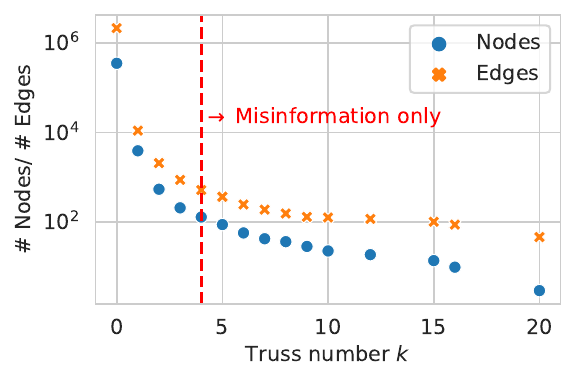}
		\caption{Number of users (nodes) and retweets (edges) in the $(k,\Delta_{50\%})$-trusses.}
		\label{fig:twitter_sizes_h_truss}
	\end{subfigure}\hfill%
	\begin{subfigure}{.23\linewidth}\centering
		\includegraphics[width=1\linewidth]{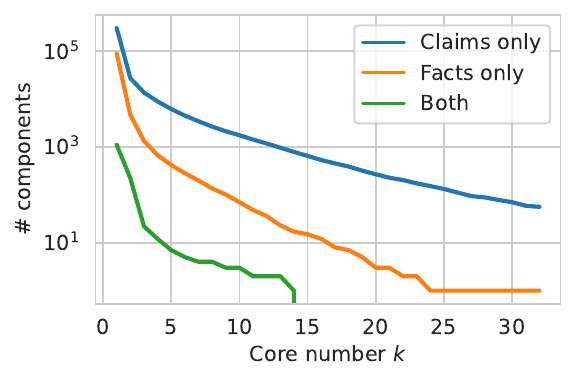}
		\caption{Numbers of $\Delta$-connected components of the $(k,\Delta_{50\%})$-cores with only claims, facts, or both.}
		\label{fig:twitter_ratios_h}
	\end{subfigure}\hfill%
	\begin{subfigure}{.23\linewidth}\centering
		\includegraphics[width=1\linewidth]{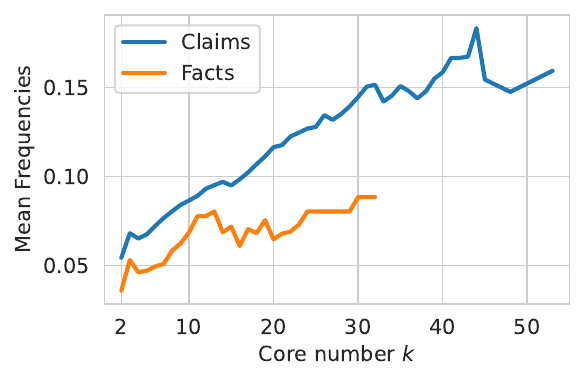}
		\caption{Mean frequencies $(s^{-1})$ of retweets in $\Delta_{50\%}$-connected components.}%
	\label{fig:twitter_freqencies_h}
\end{subfigure}   
\caption{Statistics of the $(k,\Delta_{50\%})$-cores, $(k,\Delta_{50\%})$-trusses, and $\Delta_{50\%}$-connected components of the $(k,\Delta_{50\%})$-cores.}
\label{fig:twitter}
\end{figure*}

\subsection{Choosing the Parameter $\Delta$}\label{sec:delta} 
We choose the values of $\Delta$  based on the node-based inter-event times (IETs). The IETs are defined as the set $\mathcal{I}=\{t_2-t_1\mid e_1=(\{u,v\},t_1),e_2=(\{u,w\},t_2)\in \tge, t_1\leq t_2 \text{ and } e_1,e_2 \text{ are consecutive}\}$ where two temporal edges are consecutive iff.~there is no other edge $e'=(\{u,x\},t')$ with $t_1<t'< t_2$.
Let $i_1<i_2\in\mathcal{I}$ such that there exists no other $j\in\mathcal{I}$ with $i_1< j < i_2$. Then, for 
$\Delta_1,\Delta_2\in\mathbb{N}$ with $i_1<\Delta_1\leq \Delta_2 < i_2$ the $(k,\Delta_1)$-core and $(k,\Delta_2)$-core (or $(k,\Delta_1)$-truss and $(k,\Delta_2)$-truss, resp.) are isomorphic for all $k$ because the $\Delta$-degrees (or $\Delta$-supports) do not change going from $\Delta_1$ to $\Delta_2$.
Based on this observation, we cover the range of the most relevant $\Delta$ values by choosing $\Delta$ to be the $10$, $25$, $50$, and $75$-percentiles of the IETs, denoted with $\Delta_{10\%}$, $\Delta_{25\%}$, $\Delta_{50\%}$, and $\Delta_{75\%}$, respectively.
\Cref{table:props} shows the durations as well as the average and maximum $\Delta$-degree and $\Delta$-support, respectively.
As expected, with increasing $\Delta$, the maximum and average $\Delta$-degrees and $\Delta$-supports increase. 

\subsection{Efficiency}\label{sec:experiments:efficiency}
\Cref{table:kdcorerunningtimes2} shows the running times in seconds.
For low values of $\Delta$, the running times of our algorithms are comparable, or even less, than the static-$k$-core.
The running times for computing the $(k,\Delta)$-core and $(k,\Delta)$-truss decompositions increase with increasing~$\Delta$ because the numbers of $\Delta$-incident edges (triangles, resp.) increase (see \Cref{table:props}) leading to higher numbers of updates of the $\Delta$-degree, or $\Delta$-support, in the peeling-step of the algorithms.
This effect is particularly pronounced for the \emph{Twitter} data set and the increase from $\Delta_{50\%}$ to $\Delta_{75\%}$ where the maximum $\Delta$-degree and the maximum $\Delta$-support increase in two orders of magnitude leading also to a significant increase in the running time. 
Similarly, for \emph{Bitcoin}, the massive increase of the average and maximum $\Delta$-support from $\Delta_{50\%}$ to $\Delta_{75\%}$  causes \texttt{$(k,\Delta)$-Truss} to exceed the time limit.

For the $(k,h)$-core baseline (\texttt{$(k,h)$-C}), the running times decrease for increasing $h$ because only temporal edges with at least $h$ parallel temporal edges are considered in the computation. 
The $(\eta,k)$-pseudocore (\texttt{PC-$\eta$}) has in most cases (significant) higher running times compared to our algorithm and cannot finish the computations in the given time limit of twelve hours for the \emph{Bitcoin} data set and \emph{Reddit} for $\eta=32$.
The span-truss implementation (\texttt{SpanT}) is for all data sets significantly slower than \texttt{$(k,\Delta)$-Truss} and can only finish the computation in the time limit for the three data sets with the shortest total span \emph{FacebookMsg}, \emph{Enron}, and \emph{Twitter}. The reason is that the running time is in $\mathcal{O}(|\mathcal{T}|^2\cdot m^{1.5})$, i.e., has a quadratic term in the interval spanned by the network.
The \texttt{$(L,K)$-C} algorithm performs very well in terms of running time.
However, \texttt{$(L,K)$-C} \emph{fails to identify any (non-trivial) cores (or trusses, resp.)} because it is too strict by requiring a core to exist in each time step of the interval, which can be impossible in temporal networks that are sparse in each time step. For example, the average number of edges in a time step of \emph{FacebookMsg} is only slightly above one.
Due to their high memory demands \texttt{SpanC} and \texttt{$(l,\delta)$-C} fail the computations due to out-of-memory error for all data sets. The reason is that the implementations rely on representing the temporal networks with data structures that need $\mathcal{O}(|V|^2\cdot x)$ space, where $x$ is the number of timestamps $x=|T(\tg)|$ in the case of \texttt{SpanC} or duration of the interval spanned by the network, i.e., $x=|\mathcal{T}|$, for \texttt{$(l,\delta)$-C}.
In contrast, computing the $(k,\Delta)$-cores and $(k,\Delta)$-trusses is memory-efficient---\Cref{table:memory} shows the maximum memory usage of our algorithms for the large data sets. The memory usage of {{\texttt{$(k,\Delta)$-Core}}} and {{\texttt{$(k,\Delta)$-Truss}}} were equal for all data sets and in order of the numbers of temporal edges. The algorithms needed the most memory for the \emph{Bitcoin} data set with 42.7 GB.
For the small data sets the memory usage is far below one gigabyte, reaching a maximum for \emph{Twitter} of 769 MB.
\begin{table}[htb]
\centering
\caption{Memory usage of {{\texttt{$(k,\Delta)$-Core}}} and {{\texttt{$(k,\Delta)$-Truss}}}.}  
\label{table:memory}
\resizebox{0.7\linewidth}{!}{ \renewcommand{\arraystretch}{1} \setlength{\tabcolsep}{14pt}
	\begin{tabular}{lc}\toprule
		\textbf{Data~Set}    & \textbf{Max. memory usage in GB}  \\
		\midrule       
		\emph{Wikipedia}     & 13.7 \\
		\emph{StackOverflow} & 16.7 \\ 
		\emph{Reddit}        & 29.3 \\
		\emph{Bitcoin}       & 42.7 \\
		\bottomrule
	\end{tabular}
}
\end{table}

\subsection{Use Case: Analyzing Malicious Retweets}\label{sec:usecase}
We use the $(k,\Delta)$-core and $(k,\Delta)$-truss decompositions to analyze malicious tweets in the X network (formerly Twitter).
Recent works showed that spreaders of misinformation and fake news often can be found in the inner cores of the static $k$-core decomposition~\cite{shin2016corescope,shao2018anatomy}. \citet{shao2018anatomy} analyzed %
a subgraph of the Twitter graph (data set \emph{Twitter} in \Cref{table:datasets_stats2}) representing users and retweets in the critical period of six months before the 2016 US Presidential Elections. 
The network uses a fine-grained time scale of seconds, and all edges are either labeled as \emph{fact-checking} or \emph{misinformation} (or claims), where \emph{misinformation} constitutes 81.9\% of the edges.
\citet{shao2018anatomy}
showed that the users with a high (static) core value exhibit a strong tendency to disseminate misinformation and fake news. 
Here, we extend and improve the results of~\citet{shao2018anatomy} by analyzing the network using our $(k,\Delta)$-decompositions:
We computed the $(k,\Delta_{50\%})$-core and -truss decompositions and for each core or truss the $\Delta_{50\%}$-connected components.
\Cref{fig:twitter_sizes_h} shows for each $(k,\Delta_{50\%})$-core the numbers of nodes and edges; 
similarly, \Cref{fig:twitter_sizes_h_truss} shows for each $(k,\Delta_{50\%})$-truss the numbers of nodes and edges. 
The dashed vertical lines mark the value of $k$ from which on the $(k,\Delta_{50\%})$-cores and $(k,\Delta_{50\%})$-trusses contain only claims and misinformation but no fact-checking.
We observe that edges in the inner cores only represent misinformation, which not only aligns with \citet{shao2018anatomy} observation that users in inner cores tend to spread misinformation but strengthens it with the observation that the communication in inner $(k,\Delta_{50\%})$-cores, for $k\geq 33$, only consists of misinformation and does not contain any fact-checking. We see this behavior even stronger for the $(k,\Delta)$-trusses: For $k\geq 4$, the $(k,\Delta)$-trusses only consist of misinformation retweets.
Furthermore, the $\Delta_{50\%}$-connected components of the cores and trusses are highly edge-homophilic, i.e., for most cores and trusses all edges in a $\Delta$-connected component are either misinformation or fact-checking. 
\Cref{fig:twitter_ratios_h} shows the numbers of the $\Delta_{50\%}$-connected components containing only claims, fact-checking, or both. The majority of $\Delta_{50\%}$-connected components only contain claims and significantly less components only contain fact-checking or both. Moreover, for $k\geq 14$, there are no $\Delta_{50\%}$-connected components containing both.

\Cref{fig:twitter_sizes_h} shows that the numbers of users decreases for increasing $k$. Furthermore, many of the $\Delta$-connected components in the inner $(k,\Delta)$-cores and -trusses, i.e., with large $k$, consist of only a few users. In these user groups that spread misinformation, highly frequent retweeting appears.
\Cref{fig:twitter_freqencies_h} shows the mean frequencies (in $s^{-1}$) of retweets in the $\Delta$-connected components that either only contain fact-checking or misinformation spreading. For the latter, the mean frequencies are significantly higher and show a stronger increase for increasing $k\geq 12$. 

In comparison, using the baselines, we cannot to obtain these insights.
The static $k$-cores in the graph are single connected components for $k\geq 3$ preventing insights into smaller components of the network. 
\texttt{SpanC} and \texttt{$(l, \delta)$-C} run out of memory. For \texttt{SpanT} and \texttt{$(L,K)$-C}, the results are empty because of the high time-granularity of the network, resulting in too sparse time steps for these approaches. 
Using the \texttt{$(k,h)$-C} baseline, we can obtain inner cores that have a higher ratio of misinformation. But we cannot achieve a strict distinction of finding only misinformation-spreading cores, as with the $(k,\Delta)$-core decomposition, because all $(k,h)$-cores contain misinformation and fact-checking for $h\in\{2,4,8\}$. 
In the case of \texttt{PC-$\eta$} baseline, the highest core values belong to nodes that are mainly involved in spreading misinformation. However, the subgraphs corresponding to the $(\eta,k)$-cores span the complete time interval of the network, so we cannot obtain any information about the temporal properties as we do with the new $(k,\Delta)$-decompositions. 

In conclusion, our analysis underscores the intrinsic value of $(k,\Delta)$-decompositions, especially when combined with $\Delta$-connected components, in identifying temporal structures and dynamic patterns within fine-grained temporal networks. 
Finally, see \Cref{appendix:corevstruss} for an additional comparison of the cohesiveness of the $(k,\Delta)$-core and $(k,\Delta)$-truss decompositions.

\section{Conclusion and Future Work}\label{sec:conclusion} 

We addressed the hierarchical decomposition of temporal networks by introducing a novel edge-based framework that incorporates the temporal dimension, leading to the development of the $(k,\Delta)$-core and $(k,\Delta)$-truss decompositions. Our highly efficient algorithms successfully handled large-scale, dynamic temporal networks where existing methods failed, and we demonstrated their effectiveness in a real-world use case of analyzing malicious Twitter content.

Currently, our algorithms require a predefined temporal distance~$\Delta$. Although it's straightforward to compute decompositions for all relevant $\Delta$ values, future work will focus on developing more efficient algorithms for simultaneous decomposition in both $k$ and $\Delta$. Additionally, we plan to explore size-restricted $(k,\Delta)$-cores or trusses and consider advanced edge-weighting functions.

\section*{Acknowledgments}
Giuseppe F. Italiano was partially supported by the Italian Ministry of University and Research under PRIN Project n. 2022TS4Y3N - EXPAND: scalable algorithms for EXPloratory Analyses of heterogeneous and dynamic Networked Data.

\clearpage %

\onecolumn
\appendix
\section*{Appendix}

\section{Omitted Proofs}\label{appendix:proofs}

\begin{proof}[Proof of \Cref{theorem:general}]
	After initialization $c[e]=\varphi(e,\tge)$ for all $e\in\tge$.
	Let $e_1,\ldots,e_m$ be the sequence of temporal edges removed from the network.
	First, $e_1$ is one of the edges with $c[e_1]$ being the lowest value $r_1$ and $e_1$ is in the $\tcore^{\varphi}_{r_1}$ component of $\tg$ but not in $\tcore^{\varphi}_{r'}$ with $r'>r_1$.
	Hence, $c[e_1]=c_\varphi(e_1)$ and after removing edge $e_1$ for all remaining edges $e_i$ with $2\leq i \leq m$ it holds $c[e_i]\geq c_\varphi(e_1)$.
	Now, consider the temporal edge $e_j$, processed in the $j$-th iteration of the while-loop, i.e., $c[e_j]$ is minimal for all remaining edges, i.e., all for all remaining edges in $e'\in\tge'$ it is $c[e']\geq r_j$. Hence,
	$e_j$ is in the $\tcore^{\varphi}_{r_j}$ component of $\tg$ but not in $\tcore^{\varphi}_{r'}$ with $r'>r_j$, and consequently, $c[e_j]=c_\varphi(e_j)$. Because, $c[e]$ is not updated after temporal edge $e\in\tge$ is processed, the algorithm returns the correct $c[e]=c_\varphi(e)$ for all $e\in\tge$.
	
	Finally, due to the monotonicity of $\varphi$, the order in which any two temporal edges $h,h'\in\tge$ with $c_\varphi(h)=c_\varphi(h')$ are removed does not affect the solution.
	Assume it would matter, i.e., there are two sequences of removals $\sigma_1=(e_1,\ldots,e_k,h,\ldots,h',\ldots)$ and $\sigma_2=(e_1,\ldots,e_k,h',\ldots,h,\ldots)$ with $c_\varphi(h)=c_\varphi(h')$ such that $c_1=c[h]$ wrt.~to $\sigma_1$ differs from $c_2=c[h]$ wrt.~$\sigma_2$.
	After removing the edges up to $e_k$, $c[h]=c[h']$, otherwise the order of removing them is fixed by the algorithm.
	Assume we remove $h$ first, then all remaining edges $e\in\tge'$ still have $c[e]\geq c[h]$. However, due to monotonicity of $\varphi$ it also holds that $\varphi(e,\tge'\setminus\{f\})\leq \varphi(e,\tge')$. Hence, $c[h']$ will stay equal to $c[h]$.
	Removing $h'$ first, is analogously. Hence the algorithm leads to a unique decomposition. 

\end{proof}

\begin{proof}[Proof of \Cref{lemma:monotone}]
	Let $\tge_1\subseteq \tge_2 \subseteq \tge$ and $e\in\tge$.
	If $e\not\in\tge_2$, then $d_\Delta(e,\tge_1) = d_\Delta(e,\tge_2) = 0$.
	Otherwise, if $e\in \tge_2$ but $e\not\in\tge_1$, $d_\Delta(e,\tge_1)=0 \leq d_\Delta(e,\tge_2)$.
	And finally, if $e\in \tge_1$ it holds that
	\begin{align*}
		d_\Delta(e,\tge_1)=\min (&|\left\{(\{u,w\},t')\in \tge_1 \mid |t-t'|\leq \Delta \right\}|,\\&|\left\{(\{v,w\},t')\in \tge_1 \mid |t-t'|\leq \Delta \right\}|)
		\\\leq \min (&|\left\{(\{u,w\},t')\in \tge_2 \mid |t-t'|\leq \Delta \right\}|,\\&|\left\{(\{v,w\},t')\in \tge_2 \mid |t-t'|\leq \Delta \right\}|) =d_\Delta(e,\tge_2).
	\end{align*}
\end{proof}

\begin{proof}[Proof of \Cref{lemma:eqcorenum}]
	For $\Delta_m=\max\{\Delta\in D\}$, it holds for each pair of incident edge $(\{u,v\},t_1)$ and %
	$(\{u,w\},t_2)$ the waiting time $|t_1-t_2|$ 
	at $u$ is not larger than $\Delta_m$. Hence, the temporal restriction of \Cref{def:deltadeg} is trivially respected for all edges leading to the result.
\end{proof}

\begin{proof}[Proof of \Cref{theorem:containment}]
	Assume that for $k',\Delta'\in\mathbb{N}$, it holds that $k\leq k'$, $\Delta \geq \Delta'$  with at least one of the inequalities being strict.
	We consider the following cases:
	\begin{itemize}[leftmargin=5.5mm]
		\item $k<k'$ and $\Delta=\Delta'$: For each temporal edge $e=(\{u,v\},t)$ it holds that if $e$ is part of $\tg_{(k',\Delta)}$, both of the endpoints $u$ and $v$ are connected to at least $k'+1$ neighbors in $\tg_{(k',\Delta)}$. This also means that $u$ and $v$ are connected to at least $k+1<k'+1$ neighbors, and hence, the edge is also in $\tg_{(k,\Delta)}$.
		\item $k=k'$ and $\Delta>\Delta'$: For each temporal edge $e=(\{u,v\},t)$ it holds that if $e$ is part of $\tg_{(k,\Delta')}$, both of the endpoints $u$ and $v$ are connected to at least $k+1$ neighbors in $\tg_{(k,\Delta')}$ such that for each such incident edge $(\{w,u\},t_1)$ or $(\{v,x\},t_2)$ it has to hold that the waiting time $|t_1-t|$ at $u$, or $|t_2-t|$ at $v$, respectively,  is not larger than $\Delta'$. This also means that the waiting time is not larger than $\Delta>\Delta'$, and hence, the edge is also in $\tg_{(k,\Delta)}$.
		\item $k<k'$ and $\Delta>\Delta'$: Follows from the first two cases combined.
	\end{itemize}
\end{proof}

\begin{proof}[Proof of \Cref{theorem:kdcore}]
	Let $\{e_1,\ldots,e_m\}$ be the sequence in which the edges are processed by the for loop in line \ref{alg:coredecomp:for2}. 
	After initialization, $d[e]$ equals the number minimum of $\Delta$-incident edges at its endpoints.
	Because for $e_1$ the value of $d[e_1]$ is the smallest, it follows that $e_1$ has exactly $d[e_1]$ $\Delta$-incident edges at one of its endpoints, and each edge $e$ has at least $d[e]\geq d[e_1]$ $\Delta$-incident edges at both endpoints. Therefore, $e_1$ is part of a maximal subgraph in which each edge has at both endpoints at least $d[e_1]$ $\Delta$-incident edges, i.e., $c_\Delta(e_1) = d[e_1]$.
	In each round, the processed edge is removed, and the affected degrees $d[f]$ of $\Delta$-incident edges are updated.
	In the $i$-th round, \Cref{alg:coredecomp} processes $e_i$ with $d[e_i]$, i.e., the edge with the smallest $\Delta$-degree currently remaining in the graph.
	Now we argue again that because for $e_i$ the value of $d[e_i]$ is the smallest, $e_i$ has exactly $d[e_i]$ $\Delta$-incident edges at one of its endpoints. Furthermore, each edge $e$ remaining in $\tg$ has at least $d[e]\geq d[e_i]$ $\Delta$-incident edges at both endpoints. Therefore, $e_i$ is part of a maximal subgraph in which each edge has at both endpoints at least $d[e_i]$ $\Delta$-incident edges, i.e., $c_\Delta(e_i) = d[e_i]$. Hence, after the loop in line \ref{alg:coredecomp:for2} ends, the algorithm returns with $d$ the core numbers for all $e\in\tge$.
	
	For the running time, the initialization of $d$ is possible in $\mathcal{O}(m\cdot \delta_m)$ if the edges are stored at each vertex in chronologically ordered incident lists.
	Next, we order the edges by their $\Delta$-degree using bin-sort and update the sorting after each iteration of the for-loop in line \ref{alg:coredecomp:for2} in constant time.
	In each iteration, we may have to update the value $d[f]$ for each of the at most $\delta_m$ to $e$ $\Delta$-incident edges.
	
	For the space complexity, note that we only need the arrays $a$ and $d$, which are all in $\mathcal{O}(m)$. Furthermore, the space complexity of bin sort is in $\mathcal{O}(m)$. 
\end{proof}

\begin{proof}[Proof of \Cref{lemma:trussmonotone}]
	Let $\tge_1\subseteq \tge_2 \subseteq \tge$ and $e\in\tge$.
	If $e\not\in\tge_2$, then $s_\Delta(e,\tge_1) = s_\Delta(e,\tge_2) = 0$.
	Otherwise, if $e\in \tge_2$ but $e\not\in\tge_1$, $s_\Delta(e,\tge_1)=0 \leq s_\Delta(e,\tge_2)$.
	And finally, if $e\in \tge_1$ it holds that
	\begin{align*}
		s_\Delta(e,\tge_1)=&|\{ \{e_i,e_j\} \mid  e_i=(\{u,w\},t_1),e_j=(\{v,w\},t_2)\in\tge_1 \text{ with } u\neq v\neq w, |t-t_1|\leq\Delta, |t-t_2|\leq\Delta \\&\quad\text{ and }|t_1-t_2|\leq\Delta\}|\\
		\leq &|\{ \{e_i,e_j\} \mid  e_i=(\{u,w\},t_1),e_j=(\{v,w\},t_2)\in\tge_2 \text{ with } u\neq v\neq w, |t-t_1|\leq\Delta, |t-t_2|\leq\Delta \\&\quad\text{ and }|t_1-t_2|\leq\Delta\}|\\=&s_\Delta(e,\tge_2).
	\end{align*}
\end{proof}

\begin{proof}[Proof of \Cref{theorem:kdtruss}]
	\Cref{alg:trussdecomp} shows the algorithm.
	Let $\{e_1,\ldots,e_m\}$ be the sequence in which the edges are processed by the for loop in line \ref{alg:trussdecomp:for}. 
	After initialization, $\tau[e]$ equals the number minimum of $\Delta$-support of the edges.
	Because for $e_1$ the value of $\tau[e_1]$ is the smallest, it follows that $e_1$ has exactly $\tau[e_1]$ $\Delta$-support, and each edge $e$ has at least $\tau[e]\geq \tau[e_1]$ $\Delta$-support. Therefore, $e_1$ is part of a maximal subgraph in which each edge has a $\Delta$-support of at least $\tau[e_1]$, i.e., $s_\Delta(e_1) = \tau[e_1]$.
	In each round, the processed edge is removed, and the affected $\tau[f]$ of are updated.
	In the $i$-th round, \Cref{alg:trussdecomp} processes $e_i$ with $\tau[e_i]$, i.e., the edge with the smallest $\Delta$-support currently remaining in the graph.
	Now we argue again that because for $e_i$ the value of $\tau[e_i]$ is the smallest, $e_i$ has exactly $\tau[e_i]$ $\Delta$-support. Furthermore, each edge $e$ remaining in $\tg$ has at least $\tau[e]\geq \tau[e_i]$ $\Delta$-support. Therefore, $e_i$ is part of a maximal subgraph in which each edge has at both endpoints at least $\tau[e_i]$ $\Delta$-support, i.e., $c_\Delta(e_i) = \tau[e_i]$. Hence, after the loop in line \ref{alg:trussdecomp:for} ends, the algorithm returns with $\tau$ the truss numbers for all $e\in\tge$.
	
	For the running time, the initialization of $\tau$ is possible in $\mathcal{O}(m\cdot \xi^2)$ if the edges are stored at each vertex in chronologically ordered incident lists.
	Next, we order the edges by their $\Delta$-support using bin-sort and update the sorting after each iteration of the for-loop in line \ref{alg:trussdecomp:for} in constant time.
	In each iteration, we may have to update the value $\tau[f]$ for each of each pair in $I_\Delta$ in at most $|I_\Delta|$ time.
	Computing $I_\Delta$ is possible in $\mathcal{O}(\max(\log\delta_m,\xi^2)$ by first using binary search to find edges $\Delta$-incident to the endpoints of $e$ in $\mathcal{O}(m\cdot \log\delta_m)$ and then enumerating the relevant triangles in $\mathcal{O}(\xi^2)$ time.
	
	For the space complexity, note that we need the arrays $\tau$ and $I_\Delta$, which are in $\mathcal{O}(s_\Delta^{\max})$ and $\mathcal{O}(m)$, respectively. Furthermore, the space complexity of bin sort is in $\mathcal{O}(m)$. 
\end{proof}

\begin{proof}[Proof of \Cref{theorem:dccs}]
	Let $\mathcal{C}_k^\star$ with $\star\in\{d_\Delta,s_\Delta\}$ be a $(k,\Delta)$-core or $(k,\Delta)$-truss. 
	And, let $\mathcal{D}$ be the set of $\Delta$-connected components of $\mathcal{C}_k^\star$.
	First, assume $\mathcal{C}_k^\star$ is a $(k,\Delta)$-core, i.e., $\star=d_\Delta$.
	Let $D\in\mathcal{D}$ and $e$ be an edge in $D$. The $\Delta$-degree of $d_\Delta(e,D)$ equals the $\Delta$-degree $d_\Delta(e,\mathcal{C}_k^{d_\Delta})$ as due to the definition of $\Delta$-connectedness all edges $\tge'$ contributing to $d_\Delta(e,\mathcal{C}_k^{d_\Delta})$ are $\Delta$-reachable from $e$. Therefore, the edges $\tge'$ are also in $D$.  
	Because $D$ is subgraph of $\mathcal{C}_k^{d_\Delta}$ in which all edges have the same $\Delta$-degree as in $\mathcal{C}_k^{d_\Delta}$, $D$ is a $(k,\Delta)$-core.
	
	For the case that $\mathcal{C}_k^\star$ is a $(k,\Delta)$-truss, i.e., $\star=s_\Delta$, we argue analogously.
	Here, the $\Delta$-support of $s_\Delta(e,D)$ equals the $\Delta$-support $s_\Delta(e,\mathcal{C}_k^{d_\Delta})$ as due to the definition of $\Delta$-connectedness all edges $\tge'$ contributing to $s_\Delta(e,\mathcal{C}_k^{d_\Delta})$ are $\Delta$-reachable from $e$, leading to the result.

\end{proof}

\section{$(k,\Delta)$-Truss Decomposition Algorithm}\label{appendix:trussalg}
\Cref{alg:trussdecomp} shows our $(k,\Delta)$-truss decomposition.
After initializing $\tau[e]$ with $s_\Delta(e)$ for all edges $e\in\tge$ in line \ref{alg:trussdecomp:init} the algorithm sorts the edges in increasing order of $\tau$ using bucket sort.
Next, it iterates over the temporal edges in order of minimal current $\Delta$-support (line \ref{alg:trussdecomp:for}), removes the current edge (line \ref{alg:trussdecomp:forend}), and updates all affected edges and their positions in the processing order (lines 4-8).

\begin{algorithm2e}[htb]
	\label[algorithm]{alg:trussdecomp}
	\caption{$(k,\Delta)$-truss decomposition}
	\Input{Temporal graph $\tg=(V,\tge)$ and $\Delta\in\mathbb{N}$}
	\Output{Truss number $\tau_{\Delta}(e)$ for all $e\in \tge$}
	Initialize $\tau[e_i]=s_\Delta(e_i,\tge)$ for all $e_i\in \tge$\label{alg:trussdecomp:init} \;
	Bin sort edges $\tge$ in increasing order of $\tau[e]$\;
	\For{$e=(\{u,v\},t)\in \tge$ in sorted order}{\label{alg:trussdecomp:for}
		$I_\Delta \gets \{ \{e_i,e_j\} \mid  e_i=(\{u,w\},t_1),e_j=(\{v,w\},t_2)\in\tge \text{ with } u\neq v\neq w, |t-t_1|\leq\Delta, |t-t_2|\leq\Delta \text{ and }|t_1-t_2|\leq\Delta\}$ \;
		\For{$\{e_i,e_j\}\in I_\Delta$}{
			\lIf{$\tau[e_i]>\tau[e]$}{$\tau[e_i]\gets \tau[e_i]-1$}
			\lIf{$\tau[e_j]>\tau[e]$}{$\tau[e_j]\gets \tau[e_j]-1$}
			Update bin positions of $e_i$ and $e_j$\;
		}
		remove $e$ from $\tg$\;\label{alg:trussdecomp:forend}
	}
	\Return $\tau$
\end{algorithm2e}

\section{Computing $\Delta$-Connected Components}\label{appendix:ccs}
We propose a new linear time algorithm to decompose a temporal graph into its $\Delta$-connected components based on the following transformation of the temporal graph into a static graph. %

\begin{appendixdefinition}
	Let $\tg=(V,\tge)$ be a temporal graph and $\Delta\in\mathbb{N}$.
	We define the static $\Delta$-representation $S_\Delta(\tg)=(V_S, E_W\cup E_T)$: %
	\begin{compactenum}
		\item For $(\{u,v\},t)\in \tge$, we have the time-nodes $(u,t),(v,t)\in V_S$,
		\item for each pair of consecutive time-nodes $(u,t_1),(u,t_2)\in V_S$ with $|t_1-t_2|\leq \Delta$, we have $\{(u,t_1),(u,t_2)\}\in E_W$, and 
		\item for each $(\{u,v\},t)\in \tge$, we have $\{(u,t),(v,t)\}\in E_T$.
	\end{compactenum}
	We call two time-nodes $(u,t_1),(u,t_2)\in V_S$ {consecutive} iff.~$t_1<t_2$ and there is no other time-node $(u,t')\in V_S$ with $t_1<t'<t_2$. 
\end{appendixdefinition}

\Cref{app:tgexamplec} shows the transformation of the $(2,2)$-core shown in \Cref{app:tgexampleb} with the edges $e\in E_W$ dotted.
\Cref{app:tgexampled} shows the resulting $\Delta$-connected components which are also (non-inclusion-maximal) $(2,2)$-cores.
A key property of the transformed graph $S_\Delta(\tg)$ is that its static connected components are in a one-to-one correspondence to the $\Delta$-connected components of $\tg$.
\Cref{alg:onetoone} first computes the static $\Delta$-representation and then determines the conventional static connected components. The $\Delta$-connected components can then be obtained from the partition of $E_T$. 

\begin{appendixtheorem}\label{theorem:algonetoonecorrect}
	\Cref{alg:onetoone} correctly computes the $\Delta$-connected decomposition of $\tg$ in $\mathcal{O}(|\tge|)$ running time and space. %
\end{appendixtheorem}
Before we prove \Cref{theorem:algonetoonecorrect}, we show the following property.

\begin{appendixlemma}\label{lemma:onetoone}
	Let $\tg=(V,\tge)$ be a temporal graph and $\Delta\in\mathbb{N}$.
	There is a one-to-one correspondence between the $\Delta$-connected components of $\tg$ and the (static) connected components of $S_\Delta(\tg)$.
\end{appendixlemma}
\begin{proof}[Proof of \Cref{lemma:onetoone}]
	Let $C$ be a $\Delta$-connected component of $\tg$. Consider $S_\Delta(C)$ and note that it is a single connected component $S_C$. Furthermore, $S_C$ is a subgraph of $S_\Delta(\tg)$. In particular, $S_C$ is a connected component of $S_\Delta(\tg)$. If $S_C$ would not be a connected component, then there would be an edge $\{(u,t_1),(u,t_2)\}\in E_W$ with $|t_1-t_2|\leq \Delta$ or an edge $\{(u,t),(v,t)\}\in E_T$ in both cases such that one of the endpoints is in $V(S_C)$ and the other endpoint is not in $V(S_C)$. First, assume $(u,t_1)\in V(S_C)$ and $(u,t_2)\not\in V(S_C)$ with $|t_1-t_2|\leq \Delta$. The existence of the two time-nodes $(u,t_1)\in V(S_\Delta(\tg))$ and $(u,t_2)\in V(S_\Delta(\tg))$ with $|t_1-t_2|\leq \Delta$ implies that there exists two $\Delta$-incident temporal edges $(\{u,v\},t_1), (\{u,w\},t_2)\in\tge$. However, this would mean that either both $(u,t_1)\in V(S_C)$ and $(u,t_2)\in V(S_C)$ or $(u,t_1)\not\in V(S_C)$ and $(u,t_2)\not\in V(S_C)$.
	Similarly, if we assume that there exists an edge $\{(u,t),(v,t)\}\in E_T$ it implies that there is an edge $(\{u,v\},t)\in E(C)$, and this would mean that both $(u,t)\in V(S_C)$ and $(v,t)\in V(S_C)$.

	For the other direction, let $S$ be a connected component of $S_\Delta(\tg)$. Consider the subgraph $\tg'$ of $\tg$ that corresponds to $S$. And assume that $\tg'$ is not a $\Delta$-connected component of $\tg$. This means that there is a temporal edge $(\{u,v\},t_1)\in\tge$ which is $\Delta$-incident to an edge $(\{u,w\},t_2)\in E(\tg')$, i.e., $|t_1-t_2|\leq \Delta$.
	$S_\Delta(\tg)$ contains all four time-nodes $(u,t_1),(u,t_2),(v,t_1),(w,t_2)\in V_S$.
	Because $|t_1-t_2|\leq \Delta$, there is $\{(u,t_1),(u,t_2)\}\in E_W$.
	Moreover, by definition $\{(u,t_1),(v,t_1)\}\in E_T$ and $\{(u,t_2),(w,t_2)\}\in E_T$.
	Hence, the four time-nodes are connected in $S$, which implies that $\tg'$ is also $\Delta$-connected.
\end{proof}

\begin{proof}[Proof of \Cref{theorem:algonetoonecorrect}]
	Computing $S_\Delta(\tg)$ takes $\mathcal{O}(m)$ time. We construct at most $2m$ time nodes 
	and the sizes of $E_W$ and $E_T$ are in $\mathcal{O}(m)$.
	To compute $E_W$, we pass over the chronologically ordered edge stream. For each node $u\in V$ and at any time, we only keep the time step when we encountered it last as we only connect time-nodes that are consecutive in time.
	Then, the next time we process an edge containing $u$, we can determine if the condition $|t_1-t_2|\leq \Delta$ holds, and if so, we add the corresponding edge to $E_W$.
	Hence, the construction of $V_S$, $E_W$, and $E_T$, i.e., $S_\Delta(\tg)$, are all possible in $\mathcal{O}(m)$.
	After the construction of $S_\Delta(\tg)$, we can run, e.g., a DFS in $\mathcal{O}(m)$ to determine the connected components of $S_\Delta(\tg)$ which are in a one-to-one correspondence with the $\Delta$-connected components of $\tg$ (see \Cref{lemma:onetoone}).
\end{proof}

\begin{algorithm2e}[htb!]
	\label[algorithm]{alg:onetoone}
	\caption{$\Delta$-cc decomposition}
	\Input{Temporal graph $\tg=(V,\tge)$ and $\Delta\in\mathbb{N}$}
	\Output{Set of $\Delta$-cc~$P_\Delta^i$ partitioning $\tge$}
	Compute $S_\Delta(\tg)$\;
	Compute the static connected components $C^i=(V^i_S, E^i_W\cup E^i_T)$ for $i\in[\ell]$ of $S_\Delta(\tg)$\;
	$P_\Delta^i\gets \{(\{u,v\},t)\mid \{(u,t),(v,t)\}\in E^i_T\}$ for $i\in[\ell]$\;
	\Return $\{P_\Delta^i\}_{i\in[\ell]}$
\end{algorithm2e}

\begin{figure*}\centering
	\begin{subfigure}{0.33\linewidth}\centering
		\begin{tikzpicture}
			\node[treevertex,color=black,fill=white] (a) at (0, 1.6) [circle, draw] {$a$};
			\node[treevertex,color=black,fill=white] (c) at (1.6, 1.6) [circle, draw] {$c$};
			\node[treevertex,color=black,fill=white] (b) at (0, 0) [circle, draw] {$b$};
			\node[treevertex,color=black,fill=white] (d) at (1.6, 0) [circle, draw] {$d$};

			\draw[edge] (a) -- (b) node[midway, left] {$1,20$};
			\draw[edge] (b) -- (c) node[pos=0.25,above,sloped] {$3,8$};
			\draw[edge] (a) -- (c) node[midway, above] {$1,22$};
			\draw[edge] (c) -- (d) node[midway, right] {$6,20$};
			\draw[edge] (a) -- (d) node[pos=0.25,above,sloped] {$4,10$};
			\draw[edge] (d) -- (b) node[midway, below] {$6,23$};
		\end{tikzpicture}
		\caption{Temporal graph $\tg$.}
		\label{app:tgexamplea}
	\end{subfigure}\hspace{2cm}%
	\begin{subfigure}{0.33\linewidth}\centering
		\begin{tikzpicture}
			\node at (0.2, 1.8) {$(2,2)$-core};
			\node at (2.6, 1.8) {$(1,2)$-shell};

			\node[treevertex,color=black,fill=white] (a1) at (0, 1.2) [circle, draw] {$a$};
			\node[treevertex,color=black,fill=white] (c1) at (1.2, 1.2) [circle, draw] {$c$};
			\node[treevertex,color=black,fill=white] (b1) at (0, 0) [circle, draw] {$b$};
			\node[treevertex,color=black,fill=white] (d1) at (1.2, 0) [circle, draw] {$d$};
			
			\node[treevertex,color=black,fill=white] (a2) at (2.4, 1.2) [circle, draw] {$a$};
			\node[treevertex,color=black,fill=white] (c2) at (3.6, 1.2) [circle, draw] {$c$};
			\node[treevertex,color=black,fill=white] (b2) at (2.4, 0) [circle, draw] {$b$};
			\node[treevertex,color=black,fill=white] (d2) at (3.6, 0) [circle, draw] {$d$};

			\draw[edge] (a1) -- (b1) node[midway, left] {$1$};
			\draw[edge] (b1) -- (c1) node[pos=0.5,above,sloped] {$3,8$};
			\draw[edge] (a1) -- (c1) node[midway, above] {$1$};
			\draw[edge] (c1) -- (d1) node[midway, right] {$6$};
			\draw[edge] (d1) -- (b1) node[midway, below] {$6$};            
			\draw[edge] (a2) -- (b2) node[midway, left] {$20$};
			\draw[edge] (a2) -- (c2) node[midway, above] {$22$};
			\draw[edge] (c2) -- (d2) node[midway, right] {$20$};
			\draw[edge] (a2) -- (d2) node[pos=0.5,above,sloped] {$4,10$};
			\draw[edge] (d2) -- (b2) node[midway, below] {$23$};
		\end{tikzpicture}
		\caption{$(k,\Delta)$-cores for $\Delta=2$.}
		\label{app:tgexampleb}
	\end{subfigure}%
	
	\begin{subfigure}{0.33\linewidth}\centering
		\begin{tikzpicture}
			\node at (2.0, 1.8) {$S_\Delta(\tcore^{d_2}_{2})$};

			\node[treevertex,color=black,fill=white] (a1) at (0, 1.2) [circle, draw] {$a,1$};
			\node[treevertex,color=black,fill=white] (c11) at (1.2, 1.2) [circle, draw] {$c,1$};
			\node[treevertex,color=black,fill=white] (c13) at (2, 1.2) [circle, draw] {$c,3$};
			\node[treevertex,color=black,fill=white] (b11) at (0, 0) [circle, draw] {$b,1$};   
			\node[treevertex,color=black,fill=white] (b13) at (0.8, 0) [circle, draw] {$b,3$};   
			\node[treevertex,color=black,fill=white] (c12) at (3.2, 1.2) [circle, draw] {$c,8$};
			\node[treevertex,color=black,fill=white] (c18) at (4, 1.2) [circle, draw] {$c,6$};
			\node[treevertex,color=black,fill=white] (b12) at (2.0, 0) [circle, draw] {$b,8$};
			\node[treevertex,color=black,fill=white] (b18) at (2.8, 0) [circle, draw] {$b,6$};
			\node[treevertex,color=black,fill=white] (d1) at (4, 0) [circle, draw] {$d,6$};

			\draw[edge] (a1) -- (b1) node[midway, left] {};
			\draw[edge] (b13) -- (c13) node[pos=0.5,above,sloped] {};
			\draw[edge] (b12) -- (c12) node[pos=0.5,below,sloped] {};
			\draw[edge] (a1) -- (c1) node[midway, above] {};
			\draw[edge] (c18) -- (d1) node[midway, right] {};
			\draw[edge] (d1) -- (b18) node[midway, below] {};  
			
			\draw[edge,dotted] (c13) -- (c11) node[midway, below] {}; 
			\draw[edge,dotted] (b13) -- (b11) node[midway, below] {};  
			
			\draw[edge,dotted] (c12) -- (c18) node[midway, below] {}; 
			\draw[edge,dotted] (b12) -- (b18) node[midway, below] {}; 
			
		\end{tikzpicture}
		
		\caption{Static $2$-representation of the $(2,2)$-core.}
		\label{app:tgexamplec}
	\end{subfigure}\hspace{2cm}%
	\begin{subfigure}{0.33\linewidth}\centering
		\begin{tikzpicture}
			\node at (1.0, 1.8) {$\Delta$-ccs of the $(2,2)$-core};

			\node[treevertex,color=black,fill=white] (a1) at (0, 1.2) [circle, draw] {$a$};
			\node[treevertex,color=black,fill=white] (c11) at (1.2, 1.2) [circle, draw] {$c$};
			\node[treevertex,color=black,fill=white] (b11) at (0, 0) [circle, draw] {$b$};    \node[treevertex,color=black,fill=white] (c12) at (3.2, 1.2) [circle, draw] {$c$};
			\node[treevertex,color=black,fill=white] (b12) at (2.0, 0) [circle, draw] {$b$};
			\node[treevertex,color=black,fill=white] (d1) at (3.2, 0) [circle, draw] {$d$};

			\draw[edge] (a1) -- (b1) node[midway, left] {$1$};
			\draw[edge] (b11) -- (c11) node[pos=0.5,below,sloped] {$3$};
			\draw[edge] (b12) -- (c12) node[pos=0.5,above,sloped] {$8$};
			\draw[edge] (a1) -- (c1) node[midway, above] {$1$};
			\draw[edge] (c12) -- (d1) node[midway, right] {$6$};
			\draw[edge] (d1) -- (b12) node[midway, below] {$6$};            
		\end{tikzpicture}
		
		\caption{$2$-connected components of the $(2,2)$-core.}
		\label{app:tgexampled}
	\end{subfigure}
	\caption{Examples of the $(k,\Delta)$-core and $(k,\Delta)$-truss decompositions and $\Delta$-connected components.}
\end{figure*}

\section{Data Set Details}\label{appendix:dataset}
\begin{enumerate}%
	\item \emph{FacebookMsg:}\footnote{\label{fn:nr}\url{https://networkrepository.com/index.php}} A data set based on an online social network used by students~\cite{panzarasa2009patterns}.
	The online community at the University of California, Irvine, was designed for social interactions among the students.
	Vertices represent the students and temporal edges messages sent between users.
	\item \emph{Enron:}\footref{fn:nr} An email network between employees of a company~\cite{klimt2004enron}.
	Vertices represent employees and temporal edges emails.
	\item \emph{AskUbuntu:}\footnote{\label{fn:snap}\url{https://snap.stanford.edu/data/}} A network of interactions on the stack exchange website \emph{Ask Ubuntu}~\cite{paranjape2017motifs}.
	Edges represent answers to questions.
	\item \emph{Twitter:}\footnote{\url{https://doi.org/10.5281/zenodo.1154840}} A subgraph of the Twitter graph representing users and retweets in the critical period of six months before the 2016 US Presidential Elections~\cite{shao2018anatomy}.
	\item \emph{Wikipedia:}\footnote{\url{http://konect.cc/networks/}} This network is based on \emph{Wikipedia} pages and hyperlinks between them~\cite{mislove2009online}. The timestamps of the temporal edges represent the times when two pages were linked.
	Each temporal edge represents a retweet and is labeled as factual or malicious news.
	\item \emph{StackOverflow:}\footref{fn:snap} A network of interactions on the stack exchange website \emph{Stack Overflow}~\cite{paranjape2017motifs}.
	Edges represent answers to questions.
	\item \emph{Reddit:}\footnote{\label{fn:corn}\url{https://www.cs.cornell.edu/~arb/data/}} Temporal network of \emph{Reddit} user interactions containing temporal edges reflecting users replying to users at specific times~\cite{hessel2016science}.
	\item \emph{Bitcoin:}\footref{fn:corn} This is a network of bitcoin transactions with the timestamps being the creation time of the block on the blockchain containing the transaction~\cite{Kondor-2014-bitcoin}.
\end{enumerate}

\section{Comparison of the $(k,\Delta)$-Core and $(k,\Delta)$-Truss Decompositions}\label{appendix:corevstruss}
We qualitatively compare the $(k,\Delta)$-core and $(k,\Delta)$-truss decompositions using the \emph{Twitter} network.
\Cref{app:twitter} shows the numbers of nodes and temporal edges in the cores and trusses for $\Delta_{50\%}$ and $\Delta_{75\%}$.
On the $x$-axis is the rank of cores, from high rank, i.e., inner cores, to low rank.
Additionally, we show the clustering coefficient, i.e., the ratio of closed triangles to all (open and closed) triangles~\cite{watts1998collective}. The clustering coefficient is a popular measure of cohesiveness for cores and trusses. 
In all cases, the numbers of edges and nodes increase from high to low ranks as expected.
For the $(k,\Delta_{50\%})$- and $(k,\Delta_{75\%})$-cores, the clustering coefficient is zero for the high-ranked cores and only positive for lower-ranked cores, i.e., the innermost cores do not have closed triangles.
In contrast, for the $(k,\Delta_{50\%})$- and $(k,\Delta_{75\%})$-trusses, the clustering coefficient is high for the higher ranked cores and decreases with the rank.
This shows, as expected, that, similar to the static case, the truss variant leads to spatially more cohesive subgraphs.

\begin{figure}\centering
	\begin{subfigure}{.24\linewidth}\centering
		\includegraphics[width=0.8\linewidth]{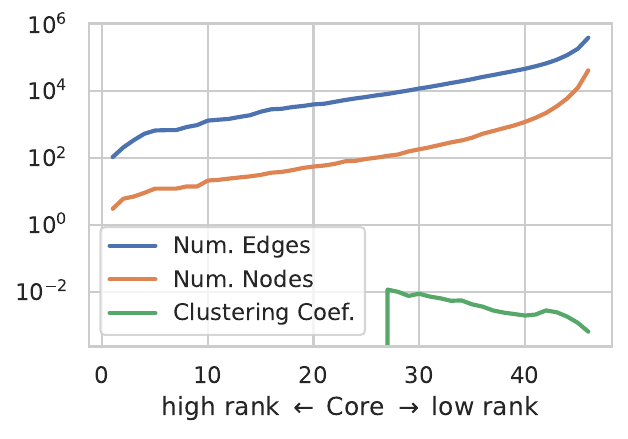}
		\caption{$(k,\Delta_{50\%})$-cores}
	\end{subfigure}\hfill%
	\begin{subfigure}{.24\linewidth}\centering
		\includegraphics[width=0.8\linewidth]{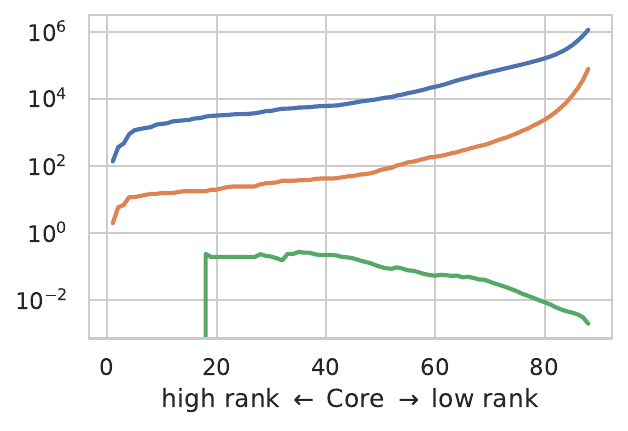}
		\caption{$(k,\Delta_{75\%})$-cores}
	\end{subfigure}\hfill%
	\begin{subfigure}{.24\linewidth}\centering
		\includegraphics[width=0.8\linewidth]{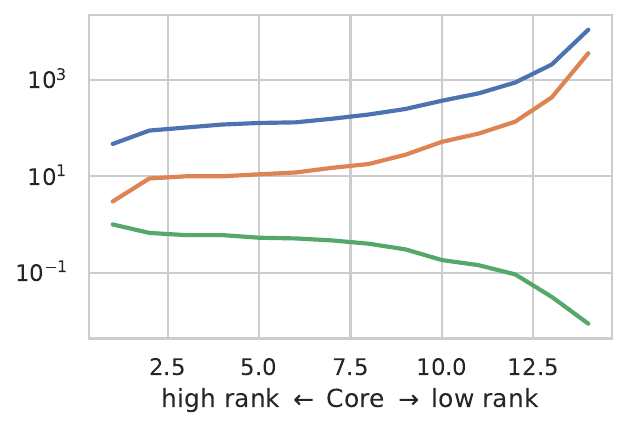}
		\caption{$(k,\Delta_{50\%})$-trusses}
	\end{subfigure}\hfill%
	\begin{subfigure}{.24\linewidth}\centering
		\includegraphics[width=0.8\linewidth]{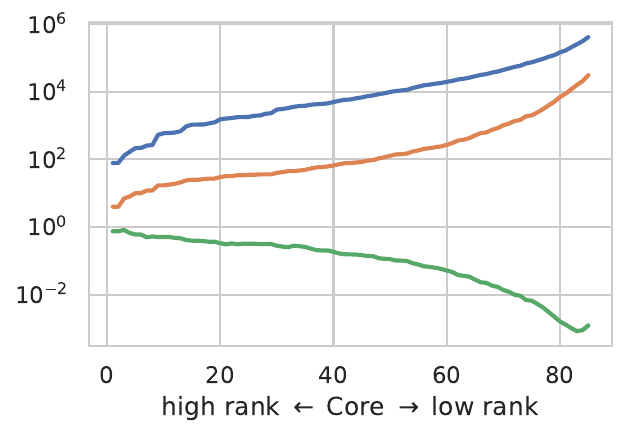}
		\caption{$(k,\Delta_{75\%})$-trusses}
	\end{subfigure}
	\vspace{-3mm}
	\caption{Sizes and clustering coefficient for the $(k,\Delta)$-cores and $(k,\Delta)$-trusses of the \emph{Twitter} data set.}
	\label{app:twitter}
\end{figure}


\begin{thebibliography}{53}
	
	%
	%
	%
	%
	%
	%
	%
	%
	%
	%
	%
	%
	%
	%
	%
	%
	
	\ifx \showCODEN    \undefined \def \showCODEN     #1{\unskip}     \fi
	\ifx \showDOI      \undefined \def \showDOI       #1{#1}\fi
	\ifx \showISBNx    \undefined \def \showISBNx     #1{\unskip}     \fi
	\ifx \showISBNxiii \undefined \def \showISBNxiii  #1{\unskip}     \fi
	\ifx \showISSN     \undefined \def \showISSN      #1{\unskip}     \fi
	\ifx \showLCCN     \undefined \def \showLCCN      #1{\unskip}     \fi
	\ifx \shownote     \undefined \def \shownote      #1{#1}          \fi
	\ifx \showarticletitle \undefined \def \showarticletitle #1{#1}   \fi
	\ifx \showURL      \undefined \def \showURL       {\relax}        \fi
	%
	%
	\providecommand\bibfield[2]{#2}
	\providecommand\bibinfo[2]{#2}
	\providecommand\natexlab[1]{#1}
	\providecommand\showeprint[2][]{arXiv:#2}
	
	\bibitem[Alvarez-Hamelin et~al\mbox{.}(2006)]%
	{alvarez2006k}
	\bibfield{author}{\bibinfo{person}{J~Ignacio Alvarez-Hamelin},
		\bibinfo{person}{Luca Dall’Asta}, \bibinfo{person}{Alain Barrat}, {and}
		\bibinfo{person}{Alessandro Vespignani}.} \bibinfo{year}{2006}\natexlab{}.
	\newblock \showarticletitle{How the k-core decomposition helps in understanding
		the internet topology}. In \bibinfo{booktitle}{\emph{ISMA Works. on Int.
			Topol.}}, Vol.~\bibinfo{volume}{1}.
	\newblock
	
	
	\bibitem[Bai et~al\mbox{.}(2020)]%
	{bai2020efficient}
	\bibfield{author}{\bibinfo{person}{Wen Bai}, \bibinfo{person}{Yadi Chen}, {and}
		\bibinfo{person}{Di Wu}.} \bibinfo{year}{2020}\natexlab{}.
	\newblock \showarticletitle{Efficient temporal core maintenance of massive
		graphs}.
	\newblock \bibinfo{journal}{\emph{Information Sciences}}  \bibinfo{volume}{513}
	(\bibinfo{year}{2020}), \bibinfo{pages}{324--340}.
	\newblock
	
	
	\bibitem[Batagelj and Zaversnik(2003)]%
	{batagelj2003m}
	\bibfield{author}{\bibinfo{person}{Vladimir Batagelj} {and}
		\bibinfo{person}{Matjaz Zaversnik}.} \bibinfo{year}{2003}\natexlab{}.
	\newblock \showarticletitle{An {O}(m) algorithm for cores decomposition of
		networks}.
	\newblock \bibinfo{journal}{\emph{arXiv cs/0310049}} (\bibinfo{year}{2003}).
	\newblock
	
	
	\bibitem[Batagelj and Zaver{\v{s}}nik(2011)]%
	{batagelj2011fast}
	\bibfield{author}{\bibinfo{person}{Vladimir Batagelj} {and}
		\bibinfo{person}{Matja{\v{z}} Zaver{\v{s}}nik}.}
	\bibinfo{year}{2011}\natexlab{}.
	\newblock \showarticletitle{Fast algorithms for determining (generalized) core
		groups in social networks}.
	\newblock \bibinfo{journal}{\emph{Advances in Data Analysis and
			Classification}} \bibinfo{volume}{5}, \bibinfo{number}{2}
	(\bibinfo{year}{2011}), \bibinfo{pages}{129--145}.
	\newblock
	
	
	\bibitem[Bhadra and Ferreira(2003)]%
	{bhadra2003complexity}
	\bibfield{author}{\bibinfo{person}{Sandeep Bhadra} {and}
		\bibinfo{person}{Afonso Ferreira}.} \bibinfo{year}{2003}\natexlab{}.
	\newblock \showarticletitle{Complexity of connected components in evolving
		graphs and the computation of multicast trees in dynamic networks}. In
	\bibinfo{booktitle}{\emph{Ad-Hoc, Mobile, and Wireless Networks:
			ADHOC-NOW2003}}. Springer, \bibinfo{pages}{259--270}.
	\newblock
	
	
	\bibitem[Bonchi et~al\mbox{.}(2019)]%
	{bonchi2019distance}
	\bibfield{author}{\bibinfo{person}{Francesco Bonchi}, \bibinfo{person}{Arijit
			Khan}, {and} \bibinfo{person}{Lorenzo Severini}.}
	\bibinfo{year}{2019}\natexlab{}.
	\newblock \showarticletitle{Distance-generalized core decomposition}. In
	\bibinfo{booktitle}{\emph{proceedings of the 2019 international conference on
			management of data}}. \bibinfo{pages}{1006--1023}.
	\newblock
	
	
	\bibitem[Casteigts et~al\mbox{.}(2022)]%
	{casteigts2022simple}
	\bibfield{author}{\bibinfo{person}{Arnaud Casteigts},
		\bibinfo{person}{Timoth{\'e}e Corsini}, {and} \bibinfo{person}{Writika
			Sarkar}.} \bibinfo{year}{2022}\natexlab{}.
	\newblock \showarticletitle{Simple, strict, proper, happy: A study of
		reachability in temporal graphs}. In \bibinfo{booktitle}{\emph{International
			Symposium on Stabilizing, Safety, and Security of Distributed Systems}}.
	Springer, \bibinfo{pages}{3--18}.
	\newblock
	
	
	\bibitem[Cohen(2008)]%
	{cohen2008trusses}
	\bibfield{author}{\bibinfo{person}{Jonathan Cohen}.}
	\bibinfo{year}{2008}\natexlab{}.
	\newblock \showarticletitle{Trusses: Cohesive subgraphs for social network
		analysis}.
	\newblock \bibinfo{journal}{\emph{National security agency technical report}}
	\bibinfo{volume}{16}, \bibinfo{number}{3.1} (\bibinfo{year}{2008}),
	\bibinfo{pages}{1--29}.
	\newblock
	
	
	\bibitem[Costa et~al\mbox{.}(2023)]%
	{costa2023computing}
	\bibfield{author}{\bibinfo{person}{Isnard~Lopes Costa}, \bibinfo{person}{Raul
			Lopes}, \bibinfo{person}{Andrea Marino}, {and} \bibinfo{person}{Ana Silva}.}
	\bibinfo{year}{2023}\natexlab{}.
	\newblock \showarticletitle{On Computing Large Temporal (Unilateral) Connected
		Components}. In \bibinfo{booktitle}{\emph{IWOCA}}. Springer,
	\bibinfo{pages}{282--293}.
	\newblock
	
	
	\bibitem[Enright and Kao(2018)]%
	{enright2018epidemics}
	\bibfield{author}{\bibinfo{person}{Jessica Enright} {and}
		\bibinfo{person}{Rowland~Raymond Kao}.} \bibinfo{year}{2018}\natexlab{}.
	\newblock \showarticletitle{Epidemics on dynamic networks}.
	\newblock \bibinfo{journal}{\emph{Epidemics}}  \bibinfo{volume}{24}
	(\bibinfo{year}{2018}), \bibinfo{pages}{88--97}.
	\newblock
	
	
	\bibitem[Galimberti et~al\mbox{.}(2020)]%
	{galimberti2020span}
	\bibfield{author}{\bibinfo{person}{Edoardo Galimberti},
		\bibinfo{person}{Martino Ciaperoni}, \bibinfo{person}{Alain Barrat},
		\bibinfo{person}{Francesco Bonchi}, \bibinfo{person}{Ciro Cattuto}, {and}
		\bibinfo{person}{Francesco Gullo}.} \bibinfo{year}{2020}\natexlab{}.
	\newblock \showarticletitle{Span-core decomposition for temporal networks:
		Algorithms and applications}.
	\newblock \bibinfo{journal}{\emph{TKDD}} \bibinfo{volume}{15},
	\bibinfo{number}{1} (\bibinfo{year}{2020}), \bibinfo{pages}{1--44}.
	\newblock
	
	
	\bibitem[Gendreau et~al\mbox{.}(2015)]%
	{gendreau2015time}
	\bibfield{author}{\bibinfo{person}{Michel Gendreau}, \bibinfo{person}{Gianpaolo
			Ghiani}, {and} \bibinfo{person}{Emanuela Guerriero}.}
	\bibinfo{year}{2015}\natexlab{}.
	\newblock \showarticletitle{Time-dependent routing problems: A review}.
	\newblock \bibinfo{journal}{\emph{Computers \& operations research}}
	\bibinfo{volume}{64} (\bibinfo{year}{2015}), \bibinfo{pages}{189--197}.
	\newblock
	
	
	\bibitem[Gionis et~al\mbox{.}(2024)]%
	{gionis2024mining}
	\bibfield{author}{\bibinfo{person}{Aristides Gionis}, \bibinfo{person}{Lutz
			Oettershagen}, {and} \bibinfo{person}{Ilie Sarpe}.}
	\bibinfo{year}{2024}\natexlab{}.
	\newblock \showarticletitle{Mining temporal networks}. In
	\bibinfo{booktitle}{\emph{Companion Proceedings of the ACM on Web Conference
			2024}}. \bibinfo{pages}{1260--1263}.
	\newblock
	
	
	\bibitem[Hessel et~al\mbox{.}(2016)]%
	{hessel2016science}
	\bibfield{author}{\bibinfo{person}{Jack Hessel}, \bibinfo{person}{Chenhao Tan},
		{and} \bibinfo{person}{Lillian Lee}.} \bibinfo{year}{2016}\natexlab{}.
	\newblock \showarticletitle{Science, askscience, and badscience: On the
		coexistence of highly related communities}. In
	\bibinfo{booktitle}{\emph{Proceedings of the international AAAI conference on
			web and social media}}, Vol.~\bibinfo{volume}{10}. \bibinfo{pages}{171--180}.
	\newblock
	
	
	\bibitem[Hogg and Lerman(2012)]%
	{hogg2012social}
	\bibfield{author}{\bibinfo{person}{Tad Hogg} {and} \bibinfo{person}{Kristina
			Lerman}.} \bibinfo{year}{2012}\natexlab{}.
	\newblock \showarticletitle{Social dynamics of digg}.
	\newblock \bibinfo{journal}{\emph{EPJ Data Science}} \bibinfo{volume}{1},
	\bibinfo{number}{1} (\bibinfo{year}{2012}), \bibinfo{pages}{5}.
	\newblock
	
	
	\bibitem[Holme(2015)]%
	{holme2015modern}
	\bibfield{author}{\bibinfo{person}{Petter Holme}.}
	\bibinfo{year}{2015}\natexlab{}.
	\newblock \showarticletitle{Modern temporal network theory: a colloquium}.
	\newblock \bibinfo{journal}{\emph{The European Physical Journal B}}
	\bibinfo{volume}{88}, \bibinfo{number}{9} (\bibinfo{year}{2015}),
	\bibinfo{pages}{234}.
	\newblock
	
	
	\bibitem[Holme et~al\mbox{.}(2004)]%
	{holme2004structure}
	\bibfield{author}{\bibinfo{person}{Petter Holme}, \bibinfo{person}{Christofer~R
			Edling}, {and} \bibinfo{person}{Fredrik Liljeros}.}
	\bibinfo{year}{2004}\natexlab{}.
	\newblock \showarticletitle{Structure and time evolution of an Internet dating
		community}.
	\newblock \bibinfo{journal}{\emph{Social Networks}} \bibinfo{volume}{26},
	\bibinfo{number}{2} (\bibinfo{year}{2004}), \bibinfo{pages}{155--174}.
	\newblock
	
	
	\bibitem[Huang et~al\mbox{.}(2014)]%
	{huang2014querying}
	\bibfield{author}{\bibinfo{person}{Xin Huang}, \bibinfo{person}{Hong Cheng},
		\bibinfo{person}{Lu Qin}, \bibinfo{person}{Wentao Tian}, {and}
		\bibinfo{person}{Jeffrey~Xu Yu}.} \bibinfo{year}{2014}\natexlab{}.
	\newblock \showarticletitle{Querying k-truss community in large and dynamic
		graphs}. In \bibinfo{booktitle}{\emph{Proceedings of the 2014 ACM SIGMOD
			international conference on Management of data}}.
	\bibinfo{pages}{1311--1322}.
	\newblock
	
	
	\bibitem[Hung and Tseng(2021)]%
	{HungT21}
	\bibfield{author}{\bibinfo{person}{Wei-Chun Hung} {and}
		\bibinfo{person}{Chih-Ying Tseng}.} \bibinfo{year}{2021}\natexlab{}.
	\newblock \showarticletitle{Maximum (L, K)-Lasting Cores in Temporal Social
		Networks}. In \bibinfo{booktitle}{\emph{DASFAA Intl Workshops}}.
	\bibinfo{address}{Cham}, \bibinfo{pages}{336--352}.
	\newblock
	\showISBNx{978-3-030-73216-5}
	
	
	\bibitem[Klimt and Yang(2004)]%
	{klimt2004enron}
	\bibfield{author}{\bibinfo{person}{Bryan Klimt} {and} \bibinfo{person}{Yiming
			Yang}.} \bibinfo{year}{2004}\natexlab{}.
	\newblock \showarticletitle{The enron corpus: A new dataset for email
		classification research}. In \bibinfo{booktitle}{\emph{ECML}}. Springer,
	\bibinfo{pages}{217--226}.
	\newblock
	
	
	\bibitem[Kondor et~al\mbox{.}(2014)]%
	{Kondor-2014-bitcoin}
	\bibfield{author}{\bibinfo{person}{D{\'a}niel Kondor},
		\bibinfo{person}{Istv{\'a}n Csabai}, \bibinfo{person}{J{\'a}nos Sz{\"u}le},
		\bibinfo{person}{M{\'a}rton P{\'o}sfai}, {and} \bibinfo{person}{G{\'a}bor
			Vattay}.} \bibinfo{year}{2014}\natexlab{}.
	\newblock \showarticletitle{Inferring the interplay between network structure
		and market effects in Bitcoin}.
	\newblock \bibinfo{journal}{\emph{New Journal of Physics}}
	\bibinfo{volume}{16}, \bibinfo{number}{12} (\bibinfo{year}{2014}),
	\bibinfo{pages}{125003}.
	\newblock
	
	
	\bibitem[Kong et~al\mbox{.}(2019)]%
	{kong2019k}
	\bibfield{author}{\bibinfo{person}{Yi-Xiu Kong}, \bibinfo{person}{Gui-Yuan
			Shi}, \bibinfo{person}{Rui-Jie Wu}, {and} \bibinfo{person}{Yi-Cheng Zhang}.}
	\bibinfo{year}{2019}\natexlab{}.
	\newblock \showarticletitle{k-core: Theories and applications}.
	\newblock \bibinfo{journal}{\emph{Physics Reports}}  \bibinfo{volume}{832}
	(\bibinfo{year}{2019}), \bibinfo{pages}{1--32}.
	\newblock
	
	
	\bibitem[Kovanen et~al\mbox{.}(2011)]%
	{kovanen2011temporal}
	\bibfield{author}{\bibinfo{person}{Lauri Kovanen}, \bibinfo{person}{M{\'a}rton
			Karsai}, \bibinfo{person}{Kimmo Kaski}, \bibinfo{person}{J{\'a}nos
			Kert{\'e}sz}, {and} \bibinfo{person}{Jari Saram{\"a}ki}.}
	\bibinfo{year}{2011}\natexlab{}.
	\newblock \showarticletitle{Temporal motifs in time-dependent networks}.
	\newblock \bibinfo{journal}{\emph{Journal of Statistical Mechanics: Theory and
			Experiment}} \bibinfo{volume}{2011}, \bibinfo{number}{11}
	(\bibinfo{year}{2011}), \bibinfo{pages}{P11005}.
	\newblock
	
	
	\bibitem[Li et~al\mbox{.}(2018)]%
	{li2018persistent}
	\bibfield{author}{\bibinfo{person}{Rong-Hua Li}, \bibinfo{person}{Jiao Su},
		\bibinfo{person}{Lu Qin}, \bibinfo{person}{Jeffrey~Xu Yu}, {and}
		\bibinfo{person}{Qiangqiang Dai}.} \bibinfo{year}{2018}\natexlab{}.
	\newblock \showarticletitle{Persistent community search in temporal networks}.
	In \bibinfo{booktitle}{\emph{ICDE}}. IEEE, \bibinfo{pages}{797--808}.
	\newblock
	
	
	\bibitem[Lin et~al\mbox{.}(2021)]%
	{lin2021mining}
	\bibfield{author}{\bibinfo{person}{Longlong Lin}, \bibinfo{person}{Pingpeng
			Yuan}, \bibinfo{person}{Rong-Hua Li}, {and} \bibinfo{person}{Hai Jin}.}
	\bibinfo{year}{2021}\natexlab{}.
	\newblock \showarticletitle{Mining Diversified Top-$ r $ r Lasting Cohesive
		Subgraphs on Temporal Networks}.
	\newblock \bibinfo{journal}{\emph{IEEE Transactions on Big Data}}
	\bibinfo{volume}{8}, \bibinfo{number}{6} (\bibinfo{year}{2021}),
	\bibinfo{pages}{1537--1549}.
	\newblock
	
	
	\bibitem[Lotito and Montresor(2020)]%
	{lotito2020efficient}
	\bibfield{author}{\bibinfo{person}{Quintino~Francesco Lotito} {and}
		\bibinfo{person}{Alberto Montresor}.} \bibinfo{year}{2020}\natexlab{}.
	\newblock \showarticletitle{Efficient Algorithms to Mine Maximal Span-Trusses
		From Temporal Graphs}.
	\newblock \bibinfo{journal}{\emph{arXiv preprint arXiv:2009.01928}}
	(\bibinfo{year}{2020}).
	\newblock
	
	
	\bibitem[Malliaros et~al\mbox{.}(2020)]%
	{malliaros2020core}
	\bibfield{author}{\bibinfo{person}{Fragkiskos~D Malliaros},
		\bibinfo{person}{Christos Giatsidis}, \bibinfo{person}{Apostolos~N
			Papadopoulos}, {and} \bibinfo{person}{Michalis Vazirgiannis}.}
	\bibinfo{year}{2020}\natexlab{}.
	\newblock \showarticletitle{The core decomposition of networks: Theory,
		algorithms and applications}.
	\newblock \bibinfo{journal}{\emph{The VLDB J.}} \bibinfo{volume}{29},
	\bibinfo{number}{1} (\bibinfo{year}{2020}), \bibinfo{pages}{61--92}.
	\newblock
	
	
	\bibitem[Mastrandrea et~al\mbox{.}(2015)]%
	{mastrandrea2015contact}
	\bibfield{author}{\bibinfo{person}{Rossana Mastrandrea}, \bibinfo{person}{Julie
			Fournet}, {and} \bibinfo{person}{Alain Barrat}.}
	\bibinfo{year}{2015}\natexlab{}.
	\newblock \showarticletitle{Contact patterns in a high school: a comparison
		between data collected using wearable sensors, contact diaries and friendship
		surveys}.
	\newblock \bibinfo{journal}{\emph{PloS one}} \bibinfo{volume}{10},
	\bibinfo{number}{9} (\bibinfo{year}{2015}).
	\newblock
	
	
	\bibitem[Michail(2016)]%
	{michail2016introduction}
	\bibfield{author}{\bibinfo{person}{Othon Michail}.}
	\bibinfo{year}{2016}\natexlab{}.
	\newblock \showarticletitle{An introduction to temporal graphs: An algorithmic
		perspective}.
	\newblock \bibinfo{journal}{\emph{Internet Mathematics}} \bibinfo{volume}{12},
	\bibinfo{number}{4} (\bibinfo{year}{2016}), \bibinfo{pages}{239--280}.
	\newblock
	
	
	\bibitem[Mislove(2009)]%
	{mislove2009online}
	\bibfield{author}{\bibinfo{person}{Alan~E. Mislove}.}
	\bibinfo{year}{2009}\natexlab{}.
	\newblock \emph{\bibinfo{title}{Online social networks: measurement, analysis,
			and applications to distributed information systems}}.
	\newblock \bibinfo{thesistype}{Ph.\,D. Dissertation}. \bibinfo{school}{Rice
		University}.
	\newblock
	
	
	\bibitem[Momin et~al\mbox{.}(2023)]%
	{momin2023kwiq}
	\bibfield{author}{\bibinfo{person}{Mahdihusain Momin}, \bibinfo{person}{Raj
			Kamal}, \bibinfo{person}{Shantwana Dixit}, \bibinfo{person}{Sayan Ranu},
		{and} \bibinfo{person}{Amitabha Bagchi}.} \bibinfo{year}{2023}\natexlab{}.
	\newblock \showarticletitle{KWIQ: Answering k-core Window Queries in Temporal
		Networks}.
	\newblock  (\bibinfo{year}{2023}).
	\newblock
	
	
	\bibitem[Nicosia et~al\mbox{.}(2013)]%
	{nicosia2013graph}
	\bibfield{author}{\bibinfo{person}{Vincenzo Nicosia}, \bibinfo{person}{John
			Tang}, \bibinfo{person}{Cecilia Mascolo}, \bibinfo{person}{Mirco Musolesi},
		\bibinfo{person}{Giovanni Russo}, {and} \bibinfo{person}{Vito Latora}.}
	\bibinfo{year}{2013}\natexlab{}.
	\newblock \showarticletitle{Graph metrics for temporal networks}.
	\newblock \bibinfo{journal}{\emph{Temporal networks}} (\bibinfo{year}{2013}),
	\bibinfo{pages}{15--40}.
	\newblock
	
	
	\bibitem[Nicosia et~al\mbox{.}(2012)]%
	{nicosia2012components}
	\bibfield{author}{\bibinfo{person}{Vincenzo Nicosia}, \bibinfo{person}{John
			Tang}, \bibinfo{person}{Mirco Musolesi}, \bibinfo{person}{Giovanni Russo},
		\bibinfo{person}{Cecilia Mascolo}, {and} \bibinfo{person}{Vito Latora}.}
	\bibinfo{year}{2012}\natexlab{}.
	\newblock \showarticletitle{Components in time-varying graphs}.
	\newblock \bibinfo{journal}{\emph{Chaos: An interdisc. journal of nonlinear
			science}} \bibinfo{volume}{22}, \bibinfo{number}{2} (\bibinfo{year}{2012}),
	\bibinfo{pages}{023101}.
	\newblock
	
	
	\bibitem[Oettershagen et~al\mbox{.}(2023)]%
	{tgh}
	\bibfield{author}{\bibinfo{person}{Lutz Oettershagen}, \bibinfo{person}{Nils~M.
			Kriege}, {and} \bibinfo{person}{Petra Mutzel}.}
	\bibinfo{year}{2023}\natexlab{}.
	\newblock \showarticletitle{A Higher-Order Temporal H-Index for Evolving
		Networks}. In \bibinfo{booktitle}{\emph{SIGKDD}}. \bibinfo{publisher}{ACM},
	\bibinfo{pages}{1770–1782}.
	\newblock
	\showISBNx{9798400701030}
	\urldef\tempurl%
	\url{https://doi.org/10.1145/3580305.3599242}
	\showDOI{\tempurl}
	
	
	\bibitem[Oettershagen and Mutzel(2022)]%
	{oettershagen2022tglib}
	\bibfield{author}{\bibinfo{person}{L. Oettershagen} {and} \bibinfo{person}{P.
			Mutzel}.} \bibinfo{year}{2022}\natexlab{}.
	\newblock \showarticletitle{Tglib: an open-source library for temporal graph
		analysis}. In \bibinfo{booktitle}{\emph{ICDM Workshops}}. IEEE,
	\bibinfo{pages}{1240--1245}.
	\newblock
	
	
	\bibitem[Oettershagen et~al\mbox{.}(2022)]%
	{oettershagen2022temporal}
	\bibfield{author}{\bibinfo{person}{Lutz Oettershagen}, \bibinfo{person}{Petra
			Mutzel}, {and} \bibinfo{person}{Nils~M Kriege}.}
	\bibinfo{year}{2022}\natexlab{}.
	\newblock \showarticletitle{Temporal walk centrality: ranking nodes in evolving
		networks}. In \bibinfo{booktitle}{\emph{Proceedings of the ACM Web Conference
			2022}}. \bibinfo{pages}{1640--1650}.
	\newblock
	
	
	\bibitem[Panzarasa et~al\mbox{.}(2009)]%
	{panzarasa2009patterns}
	\bibfield{author}{\bibinfo{person}{Pietro Panzarasa}, \bibinfo{person}{Tore
			Opsahl}, {and} \bibinfo{person}{Kathleen~M. Carley}.}
	\bibinfo{year}{2009}\natexlab{}.
	\newblock \showarticletitle{Patterns and dynamics of users' behavior and
		interaction: Network analysis of an online community}.
	\newblock \bibinfo{journal}{\emph{J. Am. Soc. for Inf. Sc. and Tech.}}
	\bibinfo{volume}{60}, \bibinfo{number}{5} (\bibinfo{year}{2009}),
	\bibinfo{pages}{911--932}.
	\newblock
	
	
	\bibitem[Paranjape et~al\mbox{.}(2017)]%
	{paranjape2017motifs}
	\bibfield{author}{\bibinfo{person}{Ashwin Paranjape},
		\bibinfo{person}{Austin~R. Benson}, {and} \bibinfo{person}{Jure Leskovec}.}
	\bibinfo{year}{2017}\natexlab{}.
	\newblock \showarticletitle{Motifs in temporal networks}. In
	\bibinfo{booktitle}{\emph{Proceedings of ACM WSDM}}.
	\bibinfo{pages}{601--610}.
	\newblock
	
	
	\bibitem[Qin et~al\mbox{.}(2020)]%
	{qin2020mining}
	\bibfield{author}{\bibinfo{person}{Hongchao Qin}, \bibinfo{person}{Rong-Hua
			Li}, \bibinfo{person}{Guoren Wang}, \bibinfo{person}{Xin Huang},
		\bibinfo{person}{Ye Yuan}, {and} \bibinfo{person}{Jeffrey~Xu Yu}.}
	\bibinfo{year}{2020}\natexlab{}.
	\newblock \showarticletitle{Mining stable communities in temporal networks by
		density-based clustering}.
	\newblock \bibinfo{journal}{\emph{IEEE Transactions on Big Data}}
	\bibinfo{volume}{8}, \bibinfo{number}{3} (\bibinfo{year}{2020}),
	\bibinfo{pages}{671--684}.
	\newblock
	
	
	\bibitem[Qin et~al\mbox{.}(2022)]%
	{qin2022mining}
	\bibfield{author}{\bibinfo{person}{Hongchao Qin}, \bibinfo{person}{Rong-Hua
			Li}, \bibinfo{person}{Ye Yuan}, \bibinfo{person}{Guoren Wang},
		\bibinfo{person}{Lu Qin}, {and} \bibinfo{person}{Zhiwei Zhang}.}
	\bibinfo{year}{2022}\natexlab{}.
	\newblock \showarticletitle{Mining Bursting Core in Large Temporal Graphs}.
	\newblock \bibinfo{journal}{\emph{Proceedings of the VLDB Endowment}}
	(\bibinfo{year}{2022}).
	\newblock
	
	
	\bibitem[Rossetti and Cazabet(2018)]%
	{rossetti2018community}
	\bibfield{author}{\bibinfo{person}{Giulio Rossetti} {and}
		\bibinfo{person}{R{\'e}my Cazabet}.} \bibinfo{year}{2018}\natexlab{}.
	\newblock \showarticletitle{Community discovery in dynamic networks: a survey}.
	\newblock \bibinfo{journal}{\emph{ACM computing surveys (CSUR)}}
	\bibinfo{volume}{51}, \bibinfo{number}{2} (\bibinfo{year}{2018}),
	\bibinfo{pages}{1--37}.
	\newblock
	
	
	\bibitem[Santoro and Sarpe(2022)]%
	{santoro2022onbra}
	\bibfield{author}{\bibinfo{person}{Diego Santoro} {and} \bibinfo{person}{Ilie
			Sarpe}.} \bibinfo{year}{2022}\natexlab{}.
	\newblock \showarticletitle{ONBRA: Rigorous Estimation of the Temporal
		Betweenness Centrality in Temporal Networks}. In
	\bibinfo{booktitle}{\emph{Proceedings of the ACM Web Conference 2022}}.
	\bibinfo{pages}{1579--1588}.
	\newblock
	
	
	\bibitem[Sarpe and Vandin(2021)]%
	{sarpe2021oden}
	\bibfield{author}{\bibinfo{person}{Ilie Sarpe} {and} \bibinfo{person}{Fabio
			Vandin}.} \bibinfo{year}{2021}\natexlab{}.
	\newblock \showarticletitle{OdeN: simultaneous approximation of multiple motif
		counts in large temporal networks}. In \bibinfo{booktitle}{\emph{Proceedings
			of the 30th ACM International Conference on Information \& Knowledge
			Management}}. \bibinfo{pages}{1568--1577}.
	\newblock
	
	
	\bibitem[Sarpe et~al\mbox{.}(2024)]%
	{sarpe2024scalable}
	\bibfield{author}{\bibinfo{person}{Ilie Sarpe}, \bibinfo{person}{Fabio Vandin},
		{and} \bibinfo{person}{Aristides Gionis}.} \bibinfo{year}{2024}\natexlab{}.
	\newblock \showarticletitle{Scalable Temporal Motif Densest Subnetwork
		Discovery}. In \bibinfo{booktitle}{\emph{Proceedings of the 30th ACM SIGKDD
			Conference on Knowledge Discovery and Data Mining}}.
	\bibinfo{pages}{2536--2547}.
	\newblock
	
	
	\bibitem[Seidman(1983)]%
	{seidman1983network}
	\bibfield{author}{\bibinfo{person}{Stephen~B Seidman}.}
	\bibinfo{year}{1983}\natexlab{}.
	\newblock \showarticletitle{Network structure and minimum degree}.
	\newblock \bibinfo{journal}{\emph{Social networks}} \bibinfo{volume}{5},
	\bibinfo{number}{3} (\bibinfo{year}{1983}), \bibinfo{pages}{269--287}.
	\newblock
	
	
	\bibitem[Shao et~al\mbox{.}(2018)]%
	{shao2018anatomy}
	\bibfield{author}{\bibinfo{person}{Chengcheng Shao}, \bibinfo{person}{Pik-Mai
			Hui}, \bibinfo{person}{Lei Wang}, \bibinfo{person}{Xinwen Jiang},
		\bibinfo{person}{Alessandro Flammini}, \bibinfo{person}{Filippo Menczer},
		{and} \bibinfo{person}{Giovanni~Luca Ciampaglia}.}
	\bibinfo{year}{2018}\natexlab{}.
	\newblock \showarticletitle{Anatomy of an online misinformation network}.
	\newblock \bibinfo{journal}{\emph{Plos one}} \bibinfo{volume}{13},
	\bibinfo{number}{4} (\bibinfo{year}{2018}).
	\newblock
	
	
	\bibitem[Shin et~al\mbox{.}(2016)]%
	{shin2016corescope}
	\bibfield{author}{\bibinfo{person}{Kijung Shin}, \bibinfo{person}{Tina
			Eliassi-Rad}, {and} \bibinfo{person}{Christos Faloutsos}.}
	\bibinfo{year}{2016}\natexlab{}.
	\newblock \showarticletitle{Corescope: Graph mining using k-core
		analysis—patterns, anomalies and algorithms}. In
	\bibinfo{booktitle}{\emph{IEEE ICDM}}. IEEE, \bibinfo{pages}{469--478}.
	\newblock
	
	
	\bibitem[Wang et~al\mbox{.}(2019)]%
	{wang2019time}
	\bibfield{author}{\bibinfo{person}{Yishu Wang}, \bibinfo{person}{Ye Yuan},
		\bibinfo{person}{Yuliang Ma}, {and} \bibinfo{person}{Guoren Wang}.}
	\bibinfo{year}{2019}\natexlab{}.
	\newblock \showarticletitle{Time-dependent graphs: Definitions, applications,
		and algorithms}.
	\newblock \bibinfo{journal}{\emph{Data Sci. and Engin.}} \bibinfo{volume}{4},
	\bibinfo{number}{4} (\bibinfo{year}{2019}), \bibinfo{pages}{352--366}.
	\newblock
	
	
	\bibitem[Watts and Strogatz(1998)]%
	{watts1998collective}
	\bibfield{author}{\bibinfo{person}{Duncan~J Watts} {and}
		\bibinfo{person}{Steven~H Strogatz}.} \bibinfo{year}{1998}\natexlab{}.
	\newblock \showarticletitle{Collective dynamics of ‘small-world’networks}.
	\newblock \bibinfo{journal}{\emph{nature}} \bibinfo{volume}{393},
	\bibinfo{number}{6684} (\bibinfo{year}{1998}), \bibinfo{pages}{440--442}.
	\newblock
	
	
	\bibitem[Wu et~al\mbox{.}(2015)]%
	{wu2015core}
	\bibfield{author}{\bibinfo{person}{Huanhuan Wu}, \bibinfo{person}{James Cheng},
		\bibinfo{person}{Yi Lu}, \bibinfo{person}{Yiping Ke}, \bibinfo{person}{Yuzhen
			Huang}, \bibinfo{person}{Da Yan}, {and} \bibinfo{person}{Hejun Wu}.}
	\bibinfo{year}{2015}\natexlab{}.
	\newblock \showarticletitle{Core decomposition in large temporal graphs}. In
	\bibinfo{booktitle}{\emph{Big Data}}. IEEE, \bibinfo{pages}{649--658}.
	\newblock
	
	
	\bibitem[Yang et~al\mbox{.}(2023)]%
	{yang2023scalable}
	\bibfield{author}{\bibinfo{person}{Junyong Yang}, \bibinfo{person}{Ming Zhong},
		\bibinfo{person}{Yuanyuan Zhu}, \bibinfo{person}{Tieyun Qian},
		\bibinfo{person}{Mengchi Liu}, {and} \bibinfo{person}{Jeffery~Xu Yu}.}
	\bibinfo{year}{2023}\natexlab{}.
	\newblock \showarticletitle{Scalable Time-Range k-Core Query on Temporal
		Graphs}.
	\newblock \bibinfo{journal}{\emph{arXiv preprint arXiv:2301.03770}}
	(\bibinfo{year}{2023}).
	\newblock
	
	
	\bibitem[Yu et~al\mbox{.}(2021)]%
	{yu2021querying}
	\bibfield{author}{\bibinfo{person}{Michael Yu}, \bibinfo{person}{Dong Wen},
		\bibinfo{person}{Lu Qin}, \bibinfo{person}{Ying Zhang},
		\bibinfo{person}{Wenjie Zhang}, {and} \bibinfo{person}{Xuemin Lin}.}
	\bibinfo{year}{2021}\natexlab{}.
	\newblock \showarticletitle{On querying historical k-cores}.
	\newblock \bibinfo{journal}{\emph{Proceedings of the VLDB Endowment}}
	(\bibinfo{year}{2021}).
	\newblock
	
	
	\bibitem[Zhong et~al\mbox{.}(2024)]%
	{zhong2024unified}
	\bibfield{author}{\bibinfo{person}{Ming Zhong}, \bibinfo{person}{Junyong Yang},
		\bibinfo{person}{Yuanyuan Zhu}, \bibinfo{person}{Tieyun Qian},
		\bibinfo{person}{Mengchi Liu}, {and} \bibinfo{person}{Jeffrey~Xu Yu}.}
	\bibinfo{year}{2024}\natexlab{}.
	\newblock \showarticletitle{A Unified and Scalable Algorithm Framework of
		User-Defined Temporal $(k,\mathcal{X})$-Core Query}.
	\newblock \bibinfo{journal}{\emph{IEEE Transactions on Knowledge and Data
			Engineering}} (\bibinfo{year}{2024}).
	\newblock
	
	
\end{thebibliography}
\end{document}